%% file: main.tex
\documentclass[12pt]{article} % Anonymized submission
\usepackage{matharticle}
\usepackage{etex}

\usepackage[l2tabu, orthodox]{nag}

\usepackage{xspace,enumerate}

\usepackage[dvipsnames]{xcolor}

\usepackage[T1]{fontenc}
\usepackage[full]{textcomp}

\usepackage[american]{babel}

\usepackage{mathtools}

\usepackage[
letterpaper,
top=1in,
bottom=1in,
left=1in,
right=1in]{geometry}

\usepackage{newpxtext} % T1, lining figures in math, osf in text
\usepackage{textcomp} % required for special glyphs
\usepackage[varg,bigdelims]{newpxmath}
\usepackage[scr=rsfso]{mathalfa}% \mathscr is fancier than \mathcal
\usepackage{bm} % load after all math to give access to bold math
\let\mathbb\varmathbb

\usepackage{prettyref}

\newcommand{\savehyperref}[2]{\texorpdfstring{\hyperref[#1]{#2}}{#2}}

\newrefformat{eq}{\savehyperref{#1}{\textup{(\ref*{#1})}}}
\newrefformat{eqn}{\savehyperref{#1}{\textup{(\ref*{#1})}}}
\newrefformat{lem}{\savehyperref{#1}{Lemma~\ref*{#1}}}
\newrefformat{def}{\savehyperref{#1}{Definition~\ref*{#1}}}
\newrefformat{thm}{\savehyperref{#1}{Theorem~\ref*{#1}}}
\newrefformat{prog}{\savehyperref{#1}{Program~\ref*{#1}}}
\newrefformat{cor}{\savehyperref{#1}{Corollary~\ref*{#1}}}
\newrefformat{cha}{\savehyperref{#1}{Chapter~\ref*{#1}}}
\newrefformat{sec}{\savehyperref{#1}{Section~\ref*{#1}}}
\newrefformat{app}{\savehyperref{#1}{Appendix~\ref*{#1}}}
\newrefformat{tab}{\savehyperref{#1}{Table~\ref*{#1}}}
\newrefformat{fig}{\savehyperref{#1}{Figure~\ref*{#1}}}
\newrefformat{hyp}{\savehyperref{#1}{Hypothesis~\ref*{#1}}}
\newrefformat{alg}{\savehyperref{#1}{Algorithm~\ref*{#1}}}
\newrefformat{rem}{\savehyperref{#1}{Remark~\ref*{#1}}}
\newrefformat{item}{\savehyperref{#1}{Item~\ref*{#1}}}
\newrefformat{step}{\savehyperref{#1}{step~\ref*{#1}}}
\newrefformat{conj}{\savehyperref{#1}{Conjecture~\ref*{#1}}}
\newrefformat{fact}{\savehyperref{#1}{Fact~\ref*{#1}}}
\newrefformat{prop}{\savehyperref{#1}{Proposition~\ref*{#1}}}
\newrefformat{prob}{\savehyperref{#1}{Problem~\ref*{#1}}}
\newrefformat{claim}{\savehyperref{#1}{Claim~\ref*{#1}}}
\newrefformat{relax}{\savehyperref{#1}{Relaxation~\ref*{#1}}}
\newrefformat{red}{\savehyperref{#1}{Reduction~\ref*{#1}}}
\newrefformat{part}{\savehyperref{#1}{Part~\ref*{#1}}}

\newcommand{\Sref}[1]{\hyperref[#1]{\S\ref*{#1}}}

\usepackage{nicefrac}
\usepackage{microtype}

\input{sam-macros.tex}

\title{How Hard is Robust Mean Estimation?}
\author{%
  Samuel B. Hopkins\thanks{University of California, Berkeley, \protect\url{hopkins@berkeley.edu}. Supported by a UC Berkeley Miller Fellowship.}
\and
Jerry Li \thanks{Microsoft Research, \protect \url{jerrl@microsoft.com}.}
}

\begin{document}

\maketitle

\begin{abstract}%
  Robust mean estimation is the problem of estimating the mean $\mu \in \R^d$ of a $d$-dimensional distribution $D$ from a list of independent samples, an $\e$-fraction of which have been arbitrarily corrupted by a malicious adversary.
  Recent algorithmic progress has resulted in the first polynomial-time algorithms which achieve \emph{dimension-independent} rates of error: for instance, if $D$ has covariance $I$, in polynomial-time one may find $\hat{\mu}$ with $\|\mu - \hat{\mu}\| \leq O(\sqrt{\e})$.
  However, error rates achieved by current polynomial-time algorithms, while dimension-independent, are sub-optimal in many natural settings, such as when $D$ is sub-Gaussian, or has bounded $4$-th moments.

  In this work we give worst-case complexity-theoretic evidence that improving on the error rates of current polynomial-time algorithms for robust mean estimation may be computationally intractable in natural settings.
  We show that several natural approaches to improving error rates of current polynomial-time robust mean estimation algorithms would imply efficient algorithms for the small-set expansion problem, refuting Raghavendra and Steurer's small-set expansion hypothesis (so long as $\Pclass \neq \NP$).
  We also give the first direct reduction to the robust mean estimation problem, starting from a plausible but nonstandard variant of the small-set expansion problem.
 \end{abstract}

\newpage

\tableofcontents

\newpage

%\begin{keywords}%
%  robust mean estimation, small-set expansion, complexity of learning, robust statistics, spectral graph theory
%\end{keywords}
\input{content/introduction}

\input{content/preliminaries}

% Acknowledgments---Will not appear in anonymized version
\section*{Acknowledgements}
We thank Prasad Raghavendra for numerous helpful discussion regarding this paper, and in particular for the suggestion to use random walks to prove the upper bound in Section~\ref{sec:spectral-bounds}.
We also thank David Steurer for answering emergency questions about small-set expansion.
We thank anonymous reviewers for extensive feedback on an earlier version of this manuscript, and Ilias Diakonikolas for numerous suggestions which substantially improved the manuscript.

\bibliographystyle{alpha}

\bibliography{bib/mathreview,bib/dblp,bib/custom,bib/scholar}

\newpage

\input{content/spectral-bounds}

\input{content/hardness-of-resilience}

\input{content/unique-sse}

% \input{content/2-to-4}

\input{content/open-problems}

\clearpage

\appendix

\input{content/appendix}

\end{document}

%% file: sam-macros.tex
\newcommand{\SSEH}{\textsf{SSEH}}
\newcommand{\SSEHeq}{\textsf{SSEH}_{=}}
\newcommand{\GapSSEHeq}{\textsf{Gap-SSEH}_{=}}
\newcommand{\SSEHleq}{\textsf{SSEH}_{\leq}}
\newcommand{\USSEH}{\textsf{USSEH}}
\newcommand{\Pclass}{\textsf{P}}
\newcommand{\NP}{\textsf{NP}}

\newcommand{\e}{\eps}

\newcommand{\Paren}[1]{\left(#1\right)}

\newcommand{\Brac}[1]{\left[#1\right]}

\newcommand{\Abs}[1]{\left\lvert#1\right\rvert}

\newcommand{\norm}[1]{\lVert#1\rVert}
\newcommand{\Norm}[1]{\left\lVert#1\right\rVert}

\newcommand{\iprod}[1]{\langle#1\rangle}

\newcommand{\mper}{\,.}
\newcommand{\mcom}{\,,}

\newtheorem{problem}[theorem]{Problem}

\newcommand{\cD}{\mathcal D}

\newcommand{\cN}{\mathcal N}

\newcommand\MYcurrentlabel{xxx}
\newcommand{\MYstore}[2]{%
  \global\expandafter \def \csname MYMEMORY #1 \endcsname{#2}%
}
\newcommand{\MYload}[1]{%
  \csname MYMEMORY #1 \endcsname%
}
\newcommand{\MYnewlabel}[1]{%
  \renewcommand\MYcurrentlabel{#1}%
  \MYoldlabel{#1}%
}
\newcommand{\MYdummylabel}[1]{}
\newcommand{\torestate}[1]{%
  \let\MYoldlabel\label%
  \let\label\MYnewlabel%
  #1%
  \MYstore{\MYcurrentlabel}{#1}%
  \let\label\MYoldlabel%
}
\newcommand{\restatetheorem}[1]{%
  \let\MYoldlabel\label
  \let\label\MYdummylabel
  \begin{theorem*}[Restatement of \cref{#1}]
    \MYload{#1}
  \end{theorem*}
  \let\label\MYoldlabel
}
\newcommand{\restatelemma}[1]{%
  \let\MYoldlabel\label
  \let\label\MYdummylabel
  \begin{lemma*}[Restatement of \cref{#1}]
    \MYload{#1}
  \end{lemma*}
  \let\label\MYoldlabel
}
\newcommand{\restateprop}[1]{%
  \let\MYoldlabel\label
  \let\label\MYdummylabel
  \begin{proposition*}[Restatement of \cref{#1}]
    \MYload{#1}
  \end{proposition*}
  \let\label\MYoldlabel
}
\newcommand{\restatefact}[1]{%
  \let\MYoldlabel\label
  \let\label\MYdummylabel
  \begin{fact*}[Restatement of \prettyref{#1}]
    \MYload{#1}
  \end{fact*}
  \let\label\MYoldlabel
}
\newcommand{\restate}[1]{%
  \let\MYoldlabel\label
  \let\label\MYdummylabel
  \MYload{#1}
  \let\label\MYoldlabel
}

\usepackage{paralist}

%% file: content/introduction.tex
%!TEX root = ../main.tex 

\section{Introduction}

Robust mean estimation is the following basic statistical problem: given a list of $n$ samples $X_1,\ldots,X_n$ from some unknown probability distribution $D$ on $\R^d$, an unknown $\e$-fraction of which have been arbitrarily corrupted by a malicious adversary, find a vector $\hat{\mu}$ such that $\|\hat{\mu} - \E_{X \sim \cD} X \|$ is as small as possible, where (for this paper) $\| \cdot \|$ is the Euclidean norm.

Among other natural settings, robust mean estimation models estimation using data sets which contain outliers -- due to random corruptions or malicious data poisoning -- and, if $D$ is assumed to lie in some class $\mathcal C$ of distributions, estimation when nature only produces data from a distribution which is $\e$-close to some distribution in $\mathcal C$ in statistical distance.
It is the most elementary of many high-dimensional statistical estimation problems which become both statistically and computationally difficult in the presence of a small constant fraction of adversarial corruptions: robust covariance estimation, robust learning of hidden-variable models, and more.

Statisticians have studied estimation under adversarially-chosen corruptions since the 1960s, originally with the notion of ``breakdown points'' \cite{anscombe1960rejection,tukey1960,huber1964robust,tukey1975mathematics}.
However, until recently, statistically-optimal rates of error when an $\e$-fraction of data is corrupted were out of reach for computationally efficient algorithms.
For instance, if $X_1,\ldots,X_n$ are $\e$-corrupted samples from $\cN(\mu,I)$, then the estimator which outputs the \emph{Tukey median} of $X_1,\ldots,X_n$ with high probability achieves $\| \mu - \text{TukeyMedian}(X_1,\ldots,X_n) \| \leq O(\e)$, when $n \geq d/\e^2$ \cite{tukey1975mathematics}.
Unfortunately, the Tukey median is $\NP$-hard to compute in high dimensions, at least for worst-case $X_1,\ldots,X_n$ \cite{bernholt2006robust}.

Naive polynomial-time approaches, such as individually pruning $X_i$'s at large distance to the rest of $X_1,\ldots,X_n$, suffer much worse rates of error: typically they lead to estimators $\hat{\mu}$ with $\| \mu - \hat{\mu} \| \leq O(\sqrt{\e d})$, even when the uncorrputed samples come from a nice distribution, such as a Gaussian as above.
Notably, the rate of error for such estimators grows with the ambient dimension $d$.

Recently, the first dimension-independent error rates for robust mean estimation were achieved by \cite{DBLP:conf/focs/DiakonikolasKK016}.
Simultaneously and independently, \cite{DBLP:conf/focs/LaiRV16} achieved error for robust mean estimation scaling with the dimension as $O(\log d)$.
These works sparked a great deal of activity in algorithm design for robust statistics, leading to new algorithms for robust mean estimation under sparsity assumptions, robust clustering and robust learning of mixture models, robust linear regression, and more (see~\cite{li-thesis,steinhardt-thesis} for surveys of recent work).

In spite of the substantial algorithmic success, current algorithms remain statistically sub-optimal in many settings, especially with respect to the dependence of the error rate $\| \mu - \hat{\mu}\|$ on $\e$.
In this paper we are interested in the question:
\begin{quote}
  Do current polynomial-time algorithms for high-dimensional robust mean estimation achieve optimal error rates among all polynomial-time algorithms?
\end{quote}

\paragraph{Contributions}
Our main contribution is a family of reductions from several variants of the \emph{small-set expansion problem}, a close cousin of Khot's unique games problem, to robust mean estimation and related problems.
These reductions show that (a) current approaches for improving error rates of existing algorithms for robust mean estimation under natural assumptions on $D$ (such as bounded $4$-th moments) would refute Raghavendra and Steurer's small-set expansion hypothesis, and (b) any efficient algorithm improving on the error rates of current algorithms for robust mean estimation under Steinhardt, Charikar, and Valiant's \emph{resilience} assumption on $D$ (see below) would refute a strengthened version of the small-set expansion hypothesis.

Our reductions employ tools from spectral graph theory.
We reinterpret and strengthen ideas from Barak et al's proof that the $2 \rightarrow 4$ norm of a matrix is hard to approximate to any constant factor under the small set expansion hypothesis \cite{DBLP:conf/stoc/BarakBHKSZ12}.
Our reinterpretation results in a simple characterization of small sets of vectors in the spectral embedding of a small-set expander (see Section~\ref{sec:spectral-bounds}).
This characterization leads to our main results.
Along the way we dramatically simplify (and generalize)~\cite{DBLP:conf/stoc/BarakBHKSZ12}'s proof of small-set expansion hardness of $2 \rightarrow q$ norms, which may be of independent interest.

\paragraph{Beating $\sqrt{\e}$: the complexity landscape}
We turn to a more quantitative discussion of our main question.
In order for robust mean estimation to be information-theoretically solvable with nontrivial error guarantees (that is, solvable by any algorithm, irrespective of running time), some assumption must be made on the underlying distribution $D$.
A common and mild assumption is that $D$ has covariance $\Sigma \preceq I$.
In this case, robust mean estimation is possible both information-theoretically and by polynomial-time algorithms with error rate $O(\sqrt{\e})$, and (up to constants) this is information-theoretically optimal.

Better scaling with $\e$ is possible under stronger assumptions on $D$.
For instance, if $D$ has $p$-th moments bounded by a dimension-independent constant, then error rate $O(\e^{1-1/p})$ is information-theoretically achievable.
Relatedly, Steinhardt, Charikar, and Valiant \cite{steinhardt2017resilience} introduced a weaker notion: $(\sigma,\e)$-\emph{resilience}.
A distribution $D$ is $(\sigma,\e)$-resilient if every event of probability at least $(1-\e)$ has conditional mean $\mu'$ with $\|\mu' - \mu\| \leq \sigma$, where $\mu$ is the mean of $D$.
They show that $\e$-robust mean estimation is then possible with error $O(\sigma)$.

So far, no polynomial-time algorithm is known which achieves error better than $O(\sqrt{\e})$ under any resilience assumption, nor is a polynomial-time algorithm known which achieves error better than $O(\sqrt{\e})$ under a bounded $p$-th moments assumption.
Thus, a second question which motivates this paper is:
\begin{quote}
What structure in the distribution $D$ of uncorrupted samples can be exploited by polynomial-time algorithms to perform robust mean estimation with error $\e^{1/2 + \Omega(1)}$?
\end{quote}
Of course \emph{a priori} it could be that no polynomial-time algorithm has error better than $O(\sqrt{\e})$, but this is not the case.
If $D$ is Gaussian, then error $O(\e \log(1/\e))$ can be achieved in polynomial time \cite{DBLP:conf/focs/DiakonikolasKK016}.
And, if $D$ has \emph{certifiably} bounded $p$-th moments (a strengthening of $p$-th moment boundedness introduced independently by~\cite{hopkins2018mixture} and by~\cite{kothari2018robust}), then error $O(\e^{1-1/p})$ is achievable in polynomial time.
Furthermore, many natural distributions fall into the latter category: product distributions and strongly log-concave distributions, for example.

Thus, there is a nontrivial complexity landscape in robust mean estimation.
Our results point to new points of hardness in this landscape.
We show that \emph{current approaches} to robust estimation under both resilience and moment-boundedness (which in particular also solve \emph{certification} problems associated to moments and to resilience) would refute the small-set expansion hypothesis if they could be improved to error rate $\e^{1/2 + \Omega(1)}$.
And we show that \emph{any efficient algorithm} achieving error $\e^{1/2 + \Omega(1)}$ under a resilience assumption would refute a nonstandard version of the small-set expansion hypothesis (see Section~\ref{sec:unique-sse}).

\paragraph{Complexity of learning under niceness assumptions}
Typically, results on computational complexity of learning take one of three forms: (1) reduction from an NP-hard problem, (2) reduction from a problem which is believed to be average-case hard, such as planted clique or learning parities with noise, or (3) unconditional lower bounds against of a restricted class of algorithms, such as statistical query (SQ) algorithms or particular hierarchies of convex programs.

Approach (1) is appealing because it can yield lower bounds which apply to all polynomial-time algorithms based on weak and well-tested assumptions like $P \neq NP$.
Often, however, applications of approach (1) prove hardness of learning problems under input distributions which do not satisfy natural \emph{niceness conditions} -- assumptions like such as input data being drawn from a Gaussian distribution or from a distribution with bounded moments -- because they rely on embedding gadgets in the input distribution.
Algorithm designers often avoid such complexity results by assuming niceness conditions like these.

We follow approach (1) as well (subject to the small set expansion hypothesis), but our reductions produce nice input distributions, satisfying regularity conditions such as bounded moments or resilience; we therefore provide evidence from worst-case complexity that robust learning is hard \emph{even under niceness assumptions}.
The majority of technical work in this paper is devoted to showing that our reductions produce such nice distributions.

We note that approach (3), in the form of SQ lower bounds, has been investigated for the robust mean estimation problem -- see Section~\ref{sec:related-work}.

%Our results add to a relatively short list of statistical estimation problems whose computational complexity can be analyzed via traditional worst-case reductions (see Section~\ref{sec:related-work} for a more in-depth survey).
%Typically, such reductions are only able to show that learning problems are hard under ill-conditioned or pathological distributions of samples, because the reductions rely on embedding gadgets in the data.
%Since traditional PAC learning allows for any distribution of samples, this kind of reduction can prove hardness of PAC learning.

% Conventional wisdom says that statistical problems -- where inputs are samples from a relatively nice distribution -- are incompatible with worst-case complexity assumptions.
% Hence, it is more common to analyze computational complexity of statistical problems by studying restricted classes of algorithms, like statistical query (SQ) algorithms or families convex programs (such as the sum of squares hierarchy) -- see Section~\ref{sec:related-work} for further discussion.
% Our results suggest that a small amount of worst-case-ness -- such as the presence of an $\e$-fraction of adversarial samples -- is enough for worst-case complexity to kick in.

\paragraph{Open problems}
This paper only begins the study of hardness of robust estimation problems based on worst-case complexity assumptions: there is a great deal left to do!
We outline several open problems in Section~\ref{sec:open-probs}.

\input{content/related-work}

%% file: content/related-work.tex
%!TEX root = ../main.tex

\subsection{Related Work}
\label{sec:related-work}

\paragraph{Robust statistics} The study of robust statistics, and specifically robust mean estimation, was initiated by seminal work of statisticians in the 60's and 70's~\cite{anscombe1960rejection,tukey1960,huber1964robust,tukey1975mathematics}.
However, it was not until recently that efficient algorithms were discovered for robust mean estimation in high dimensions which acheive nearly optimal error guarantees~\cite{DBLP:conf/focs/DiakonikolasKK016,DBLP:conf/focs/LaiRV16,diakonikolas2017being}.
The field has since experienced an explosion of algorithmic work.
For a survey on more recent algorithmic results, see~\cite{li-thesis,steinhardt-thesis}.

\paragraph{PAC learning lower bounds}
While there is a large literature on lower bounds for distributional learning problems either from average case assumptions or applying to a restricted classes of algorithms, there are only a handful of results we are aware of which base hardness of such problems on worst case hardness assumptions~\cite{kearns1994learnability,DBLP:journals/siamcomp/GuruswamiR09,feldman2006new,DBLP:conf/focs/ApplebaumBX08,DBLP:journals/jacm/Regev09,feldman2012computational,bun2016order,bubeck2018adversarial}.
Moreover, the lower bounds tend to be proved in a PAC learning sense, where the learning problem is \emph{worst-case over distributions}.
%data distribution is allowed to be unconstrained (or almost unconstrained).
We consider a version of robust mean estimation which is worst-case over input distributions \emph{belonging to a class of nice distributions, i.e. resilient distributions or those with bounded moments.}
This amount of worst-case-ness allows us to base our results on worst-case hardness assumptions, but requires significant work in our reductions to produce such nice input distributions.

\paragraph{Computational lower bounds in robust statistics}
In the context of robust estimation, almost all known lower bounds were either against restricted classes of algorithms, notably statistical query algorithms~\cite{diakonikolas2017statistical,diakonikolas2018list,diakonikolas2019efficient}, or against specific estimators~\cite{johnson1978densest,bernholt2006robust}.

In particular, \cite{diakonikolas2017statistical} proves an SQ lower bound in the setting of robust mean estimation for Gaussian distributions suggesting that the $O(\e)$ vs $O(\e \log 1/\e)$ gap between information-theoretically optimal error rates and those of known polynomial-time algorithms is likely inherent.
For at least one of the problems we investigate -- complexity of robust mean estimation under bounded moment assumptions -- it would not be possible to prove an analogous SQ lower bound.
This is because for every fixed $p$ there is a simple (folklore) SQ \emph{algorithm} which makes $\poly(d)$ statistical queries (with $1/\poly(d)$ tolerance) and robustly estimates the mean of a distribution with bounded $p$-th moments to (information-theoretically optimal) accuracy $O(\e^{1-1/p})$.\footnote{For simplicity we ignore the dependence of the number of statistical queries and tolerance on $\e$.}
The only implementations of this algorithm we are aware of require $\exp(d)$ additional running time \emph{but the SQ framework only allows for lower bounds on the number and tolerance of queries, not on the additional running time to process the answers to those queries}.
We expect similar guarantees to be unachievable in polynomial time (especially in light of the present work), SQ lower bounds cannot evidence this.
A different approach, such as the reduction-based arguments we pursue here, is required.

The only other work we are aware of giving a reduction from small-set expansion to prove complexity of a robust learning problem is~\cite{hardt2013algorithms}, which gives lower bounds from small set expansion for the problem of identifying a low-dimensional subspace which contains a large fraction of a high-dimensional data set.
While both their work and ours show reductions from the small-set expansion problem, the works otherwise diverge on a technical level -- our reductions employ spectral graph theory, while theirs is largely combinatorial -- and the results are incomparable.
Furthermore, besides the constraint that the distribution of ``good'' samples lives on a low dimensional subspace, they enforce no additional niceness conditions.
In particular, the distribution which results from their reduction is exponentially ill-conditioned.
This stands again in contrast to the relative niceness of the distributions resulting from our reductions.

Klivans and Kothari \cite{DBLP:conf/approx/KlivansK14} show hardness of robustly learning halfspaces with respect to Gaussian data; however, they start from an average-case hardness assumption (learning sparse parities with noise) rather than a worst-case one as we do here.

%% file: content/preliminaries.tex
%!TEX root=../main.tex

\subsection{Results}

The fundamental problem of study in this paper is robust mean estimation.
At a high level, the question is as follows: given samples from a distribution $D$, a small fraction of which have been corrupted, estimate the mean of $D$ as well as possible.
There are several possible corruption models to consider.
% The model considered in~\cite{DBLP:conf/focs/DiakonikolasKK016,DBLP:conf/stoc/CharikarSV17,steinhardt2017resilience} from the upper bounds side is the following, strong notion of corruption:
% \begin{definition}[$\delta$-corruption]
% Let $D$ be a distribution over $\R^d$.
% We say that $X_1, \ldots, X_n$ is an $\delta$-corrupted set of samples from $D$ if it is generated via the following process:
% \begin{itemize}
%   \item Samples $Y_1, \ldots, Y_n$ are drawn independently from $D$.
%   \item An adversary inspects these samples, changes an $\delta$-fraction of them arbitarily, and returns them in any order.
% \end{itemize}
% \end{definition}
% \noindent
% A weaker, but also well-studied notion of corruption,
In this work, we will show lower bounds against the following (relatively weak) notion of corruption, which dates back to work of Huber in the 1960s~\cite{huber1964robust}:
\begin{definition}[$\eps$-contamination]
Let $D$ be a distribution over $\R^d$.
We say that that $X_1, \ldots, X_n$ is an $\eps$-contaminated set of samples from $D$ if the $X_i$ are drawn i.i.d. from $(1 - \eps) D + \eps N$, where $N$ is an arbitrary, unknown distribution.
\end{definition}
This model is also known as \emph{Huber's contamination model} in the robust statistics literature.
The recent efficient algorithms~\cite{DBLP:conf/focs/DiakonikolasKK016,DBLP:conf/stoc/CharikarSV17,steinhardt2017resilience} actually work for slightly stronger notions of corruption.
All of our lower bounds will be against learning from $\eps$-contaminated samples, so in particular, they are also lower bounds against learning from corrupted samples as considered in these papers.

With these definitions, we can now formally state the robust mean estimation problem.
\begin{problem}[Robust mean estimation]
\label{prob:robust-mean}
Let $D$ be a distribution with mean $\mu$.
Given $\delta$-contaminated samples from $D$, output $\widehat{\mu}$ minimizing $\| \mu - \muhat \|_2$ with high probability.
\end{problem}
We briefly note, as matter of notation, that in Problem~\ref{prob:robust-mean} and the remainder of the paper, we will use $\delta$ (rather than $\eps$, as is standard in robust statistics) to denote the fraction of corrupted samples.
This will be helpful to stay notationally consistent with the literature on small-set expansion that we heavily rely on.

Without additional assumptions on $D$, Problem~\ref{prob:robust-mean} is impossible: there is no way to distinguish between $D$ and $\delta D + (1 - \delta) N$, and since $N$ can be arbitrary, the means of these two distributions can be arbitrarily far away.
To make this problem statistically tractable, we must impose some conditions on $D$.
In this paper we will focus on two previously considered conditions, namely, bounded moments and resilience.

\subsubsection{Bounded moments}
A canonical assumption in this area is that $D$ has some number of bounded moments.
For instance, arguably the most natural assumption is that $D$ has bounded covariance.
In this case, we have efficient algorithms matching the information theoretic lower bound:
\begin{fact}[\cite{DBLP:conf/icml/DiakonikolasKK017}]
\label{fact:second-moment-eff}
Let $\cD$ be the class of distributions over $\R^d$ distribution over $\R^d$ whose covariance have spectral norm at most $1$, or equivalently, which have
\begin{equation}
\label{eq:2-moment-bound}
\E_{X \sim D} \Abs{\iprod{v, X} - \iprod{v, \E_{X \sim D} X}}^2 \leq 1 \; ,
\end{equation}
for all unit vectors $v$.
There is a polynomial time algorithm, which for all small-enough $\delta > 0$ and all $D \in \cD$, given a $\delta$-contaminated set of samples from $D$ of size $\poly (d, 1 / \delta)$, outputs $\widehat{\mu}$ which with probability at least $9 / 10$, satisfies $\Norm{\widehat{\mu} - \E_{X \sim D} \mu } \leq O(\sqrt{\delta})$.
Moreover, no estimator (efficient or not) can achieve $\Norm{\widehat{\mu} - \E_{X \sim D} X } \leq o(\sqrt{\delta})$ with probability greater than $1/10$ for all $X \in \cD$.
\end{fact}
\noindent
Thus in this case, up to constants, there is no gap between the robustness of efficient and inefficient estimators.
An obvious question is whether this can be strengthened by making additional structural assumptions on the data.
For instance, what if we assume $p$ bounded moments, for $p > 2$?
Indeed, in this setting something stronger is possible, at least with exponential running time:
\begin{fact}[folklore]
\label{fact:pth-moment-ineff}
Let $p > 2$ and let $\cD_p$ be the class of distributions over $\R^d$ whose $p$-th central moment is at most $1$: that is, $D \in \cD_p$ if and only if
\begin{equation}
\label{eq:moment-bound}
\E_{X \sim D} \Abs{\iprod{v, X} - \iprod{v, \E_{X \sim D} X }}^p \leq 1
\end{equation}
for all unit vectors $v$,
There exists an exponential-time algorithm which for all $\delta > 0$ sufficiently small and all $D \in \cD_p$, given a $\delta$-contaminated set of samples from $D$ of size $\poly (d, 1 / \delta)$, outputs $\widehat{\mu}$ so that $\Norm{\widehat{\mu} - \E_{X \sim D} X} \leq O(\delta^{1-1/p})$ with probability at least $9 / 10$.
Moreover, no estimator achieves error $\Norm{\widehat{\mu} - \E_{X \sim D} X} < o(\delta^{1 - 1/p})$ with probability at least $1/10$ over all of $\cD_p$.
\end{fact}
In particular, Fact~\ref{fact:pth-moment-ineff} says that for $p > 2$, it is possible to outperform the guarantees of the algorithm in Fact~\ref{fact:second-moment-eff} asymptotically as $\delta \to 0$.
However, despite much work in the area, no efficient algorithms are known which achieve error better than $O(\delta^{1/2})$, i.e. the rate in the $p = 2$ case, unless even stronger assumptions are made.
This leads to the question:
\begin{question}
\label{q:moments}
Is there some $p > 2$ and a polynomial-time algorithm which for all sufficiently-small $\delta > 0$ and all $D \in \cD_p$ finds $\widehat{\mu}$ satisfying $\|\widehat{\mu} - \E_{X \sim D} X\| < o(\sqrt \delta )$ with probability at least $9/10$ when given a $\delta$-corrupted set of $\poly(d,1/\delta)$ samples from $D$?
\end{question}

Towards answering Question~\ref{q:moments}, we offer evidence that current techniques to algorithmically exploit moment boundedness cannot be extended to positively answer Question~\ref{q:moments}.
The algorithms which achieve the guarantees in Fact~\ref{fact:second-moment-eff} solve, as a subroutine, the problem of maximizing the left-hand side of \eqref{eq:2-moment-bound} over all unit $v$.
The algorithms of \cite{hopkins2018mixture,kothari2018robust} which exploit $p$-th moment boundedness when the $p$-th moments satisfy additional structural assumptions analogously require subroutines which certify upper bounds on the left-hand side of \eqref{eq:moment-bound}.

A theorem of Barak et al. on hardness of computing the $2\rightarrow q$ norm of a matrix already shows that this approach cannot be extended to $p \geq 4$ under only the assumptions specified in Question~\ref{q:moments} without violating the small-set expansion hypothesis.
In the following, $D$ should be thought of as the uniform distribution over the vectors $a_1,\ldots,a_n$.

\begin{theorem}[\cite{DBLP:conf/stoc/BarakBHKSZ12}]
\label{thm:barak-hardness}
  If for any even $q \geq 2$ there is a polynomial-time algorithm which given $a_1,\ldots,a_n \in \R^d$ outputs a constant-factor approximation to $\max_{\|x\|=1} \frac 1 n \sum_{i=1}^n \iprod{a_i, x}^q$, then there is a polynomial-time algorithm for small set expansion.
\end{theorem}

In this work we strengthen Barak et al.'s result in several ways.
Barak et al.'s result shows that for $c,s$ with $c/s$ arbitrarily large it is SSE-hard to distinguish a distribution with $4$-th moment at least $c$ from one with $4$-th moment at most $s$.
In statistical settings, however, it is natural to assume niceness of many moments.
For instance: is it possible to distinguish a distribution $D$ all of whose $q$-th moments for $q \leq 100$ have sub-Gaussian-type behavior (i.e. growing like $q^{q/2}$) from one whose $4$-th moment is very large?
An algorithm which could solve this decision problem seems likely to lead to an algorithm to improve on error $o(\sqrt{\delta})$, at least under the assumption that $D$ has $100$ sub-Gaussian moments.

We show that this apparently easier decision problem is still SSE hard.
This requires modifying Barak et al.'s reduction so that in one case a distribution with sub-Gaussian moments is obtained; we do this by composing the reduction of Barak et al. with a smoothing/averaging step which we analyze via Rosenthal's moment inequality.
The result addresses an open problem of Jacob Steinhardt \cite{steinhardt-talk}.
Additionally, we extend Barak et al.'s result to the case $p = 2 + \gamma$ for arbitrarily small $\gamma$, and we substantially simplify their proof.

%First of all, for existing algorithmic approaches to achieve error $o(\sqrt{\delta})$ in Question~\ref{q:moments} it would likely suffice to efficiently certify that the $p$th moment is bounded by $1$ (as in \eqref{eq:moment-bound}) under the condition that also the $100p$-th moment is bounded by $1$.
%a $p$-th moment optimization problem (such as maximizing the left side of \eqref{eq:moment-bound}) only under the assumption that if $\E|\iprod{X-\mu,v}|^p \leq 1$ then also $\E|\iprod{X-\mu,v}|^{100p} \leq 1$ -- that is, either $X$ has $100p$ moments bounded or the $p$-th moment is much greater than $1$.

%Barak et al.'s result does not directly rule out the possibility of such an algorithm.
%Our result says that an algorithm even for this easier problem would also refute the small-set expansion hypothesis.
%This addresses a conjecture of Steinhardt that $4$-th moment optimization problems can be solved efficiently for sub-Gaussian distributions

\begin{theorem}[Informal, see Theorem~\ref{thm:moment-cert}]
\label{thm:moments-intro}
  For any $p > 2$ and $q \in (2,p]$ and $c > s > s_0$ for some universal contant $s_0$, a polynomial time algorithm to distinguish the following two cases would yield a polynomial-time algorithm for the small-set expansion problem.
  Given $a_1,\ldots,a_n \in \R^d$, distinguish between:
  \textbf{yes:} there is a unit $x$ that $\frac 1 n \sum_{i=1}^n |\iprod{a_i,x}|^q > (cq)^{q/2}$, and \textbf{no:} for all unit $x$ and $q \leq p$ it holds that $\frac 1 n \sum_{i=1}^n |\iprod{a_i,x}|^{q} \leq (sq)^{q/2}$.
\end{theorem}

\subsubsection{Resilience}
Another recently introduced assumption is that of \emph{resilience:}
\begin{definition}[Resilience, see \cite{steinhardt-thesis}]
  Let $X$ be an $\R^d$-valued random variable with mean $\E X = \mu$.
  $X$ is $(\sigma,\delta)$-\emph{resilient} in a norm $\| \cdot \|$ if for all events $A$ with $\Pr A \geq 1-\delta$, we have
  $\Norm{ \E X \, | \, A - \mu } \leq \sigma$.
  Equivalently, $X$ is $(\sigma,\delta)$-resilient if for all events $A$ with $\Pr A \leq \delta$, we have$
  \Norm{\E X \, | \, A - \mu} \leq \sigma \cdot \frac{1 - \Pr A}{\Pr A}$.
\end{definition}
In the remainder of the paper we will primarily consider the case where the norm $\| \cdot \|$ is the $\ell_2$ norm in $\R^d$, since that is the setting in which our hardness results will apply.
For the proof of equivalence, see Lemma 3 and Lemma 10 in \cite{steinhardt2017resilience}.

It is not hard to show (see Corollary~\ref{cor:resilience-to-moments}) that if $D$ has second moments bounded by $1$, then $D$ is $(\sqrt{\delta}, \delta)$-resilient for all $\delta \leq 1/2$.
Thus it might not be surprising that in this setting, we can achieve rates for robust mean estimation similar to those in Fact~\ref{fact:second-moment-eff}, at least inefficiently.
However, there is already some asymptotic gap here between what is information-theoretically achievable and what is know to be achievable in polynomial time, since $(\sqrt{\delta},\delta)$ resilience for a \emph{fixed} $\delta$ is somewhat weaker than second moments bounded by $1$.

\begin{fact}[\cite{steinhardt2017resilience}]
\label{fact:resil-sqrt}
Let $\cD_\delta$ be the class of distributions over $\R^d$ which are $(\sqrt{\delta},\delta)$-resilient.
There exists an (exponential time) algorithm, which for all small-enough $\delta > 0$ and all $D \in \cD_\delta$, given a $\delta$-contaminated set of samples from $D$ of size $\poly (d, 1 / \delta)$, outputs $\widehat{\mu}$ which with probability at least $9 / 10$, satisfies $\Norm{\widehat{\mu} - \E_{X \sim D} X} \leq O(\sqrt{\delta})$.
Furthermore, there is a polynomial-time algorithm which achieves $\Norm{\widehat{\mu} - \E_{X \sim D} X} \leq O(\sqrt{\delta \log(1/\delta)})$.
Moreover, no estimator achieves error $\Norm{\widehat{\mu} - \E_{X \sim D} X} < c \sqrt{\delta}$ with probability at least $1/10$.
\end{fact}

A reasonable strengthening of this considers the condition that $D$ is $(\sigma,\delta)$-resilient for some $\sigma \ll \sqrt{\delta}$.
The following basic fact about resilience shows that such assumptions suffice information-theoretically to achieve improved error rates.

\begin{fact}[\cite{steinhardt2017resilience}]
\label{fact:resil-info}
There is an (inefficient) algorithm which given $\poly(d,1/\delta)$ $\delta$-contaminated samples from a $(\sigma,\delta)$-resilient distribution $D$ outputs $\widehat{\mu}$ such that with probability at least $9/10$ it holds that $\Norm{\widehat{\mu} - \E_{X \sim D} X} \leq O(\sigma)$.
\end{fact}
\noindent
In particular, Fact~\ref{fact:resil-info} implies that if $\sigma \leq o(\sqrt{\delta})$, then it is information-theoretically possible to outperform even the exponential time algorithm from Fact~\ref{fact:resil-sqrt}.
This leads to the question:

\begin{question}
\label{question:resilience}
  Is there a function $\sigma(\delta)$ and a polynomial-time algorithm which for all small-enough $\delta > 0$ given $\poly(d,1/\delta)$ $\delta$-contaminated samples from any $(\sigma,\delta)$-resilient distribution $D$ can find $\widehat{\mu}$ such that $\|\widehat{\mu} - \E_{X \sim D} X \| \leq o(\sqrt \delta)$ with probability at least $9/10$?
\end{question}

We prove two theorems suggesting a negative answer to Question~\ref{question:resilience}.
The first is in a similar spirit to Theorem~\ref{thm:moments-intro}.
Existing algorithms (both efficient and inefficient) for robust mean estimation under resilience assumptions solve as a subroutine the problem of determining whether (the uniform distribution over) a set of samples is $(\sigma,\delta)$-resilient.
Thus, a potential route to design an algorithm for Question~\ref{question:resilience} is to improve existing guarantees for algorithms to check if a set of points is resilient.
We show that such improvements would violate the small-set expansion hypothesis.

\begin{theorem}[Informal, see Theorem~\ref{thm:resilience-hardness}]
\label{thm:resilience-cert-intro}
  For every sufficiently-small $s > 0$ there exists $\delta > 0$ such that an efficient algorithm for the following problem would yield an efficient algorithm for small set expansion:
  Given a set of points $a_1,\ldots,a_n \in \R^d$, distinguish between the cases \textbf{yes:} the uniform distribution on $\{a_1,\ldots,a_n\}$ is $(s\sqrt{\delta},\delta)$ resilient, and \textbf{no:} it is not $(0.4 \sqrt{\delta}, \delta)$-resilient.
\end{theorem}

Our final theorem is the first in the literature to directly attack hardness for robust mean estimation via reduction from a worst-case complexity assumption, rather than reducing to related problems like certifying moment bounds or checking resilience as in Theorem~\ref{thm:barak-hardness}, Theorem~\ref{thm:moments-intro}, and Theorem~\ref{thm:resilience-cert-intro}.
We are able to show a negative answer to Question~\ref{question:resilience} under a strengthened small-set expansion hypothesis.
Our strengthened version, which we call the \emph{unique small-set expansion hypothesis} is as follows:

\begin{hypothesis}[Unique Small-Set Expansion Hypothesis]
For every $\e > 0$ there exists $\delta > 0$ such that given a graph $G$, it is $\NP$-hard to distinguish the following cases: \textbf{no} every set $S \subseteq [n]$ of $\delta n$ vertices has expansion $\Phi_G(S) \geq 1- \epsilon$, or \textbf{yes:} there exists a set $S \subseteq [n]$ of $\delta n$ vertices in $G$ such that $\Phi_G(S) \leq \epsilon$, and every other subset $T \subseteq [n]$ of $\delta n$ vertices with $S \cap T = \emptyset$ has $\Phi_G(T) \geq 1- \epsilon$.
\end{hypothesis}

Here \emph{unique} refers to the fact that in the \textbf{yes} case, the set $S$ is the unique small nonexpanding set in $G$.\footnote{This use of ``unique'' should not be confused with Unique Games!}
While we are not aware of this strengthening being considered previously in the literature, we also do not know any algorithmic techniques which could refute it.
Hence we view the following theorem as at least a barrier to improving existing algorithms for robust mean estimation.

\begin{theorem}[Informal, see Theorem~\ref{thm:usseh-robust-mean}]
\label{thm:unique-intro}
  If Question~\ref{question:resilience} has an affirmative answer then the Unique Small Set Expansion Hypothesis is false (or $P = NP$).
\end{theorem}

It remains an interesting open problem to see if Theorem~\ref{thm:unique-intro} can be strengthened to yield an algorithm for the (vanilla) small set expansion problem.

\subsubsection{Spectral graph theory: Cheeger-style rounding for analytically sparse vectors}
Our reductions involve spectral graph theory for small-set expanders, and one of our technical contributions is to substantially simplify current understanding of a simple structural question in spectral graph theory.
This leads to the proofs of our main theorems, and answers an open question of Barak on simplification of the proof that the $2\rightarrow 4$ norm is hard to approximate under the small-set expansion hypothesis (see Exercise 6.2 in \cite{barak-open-problem}).

We review definitions formally in Section~\ref{sec:prelims}, but let us briefly recall some basics.
For a regular $n$-node graph $G$ and a set $S \subseteq [n]$, the expansion of $S$, denoted $\Phi_G(S)$, is the probability that a random walk initialized uniformly in $S$ leaves it after one step.
If we denote also by $G$ the normalized adjacency matrix, then the expansion is $\Phi_G(S) = 1 - \iprod{1_S, G 1_S} / |S|$.
Of course, this makes sense only for indicator vectors $1_S$ of sets of vertices.

Cheeger's inequality extends the relationship between the quadratic form of $G$ and expansion to other vectors.
A consequence of Cheeger's inequality is the following fact:
\begin{fact}[Consequence of Cheeger's inequality]
If $v$ is any unit vector $v$ where $\iprod{v, Gv} \geq 1/2$ (and $v$ is orthogonal to the all-$1$'s vector), there is a level set $S$ of the vector with $w_i = |v_i|$ with expansion $\Phi_G(S) \leq 0.99$.
\end{fact}

In the context of small-set expansion, it is important to detect the existence \emph{small} sets of vertices -- say, $\delta n$ vertices for small constants $\delta$ -- with expansion bounded away from $1$.
A key question is: what \emph{analytical} properties of a vector $v$ with $\iprod{v, Gv} \geq 1/2$ give rise to a set of $\delta n$ vertices $S$ with expansion $\Phi_G(S) \leq 0.99$?

\cite{DBLP:conf/stoc/BarakBHKSZ12} showed that it is sufficient for $v$ to be \emph{analytically sparse}.
In particular, they showed that if $\|v\|_4^4 \geq 1/\delta$ -- that is, the $4$-norm of $v$ is similar to that of the (scaled) indicator vector of a set of size $\delta n$, then one may find a set of $\delta n$ vertices in $G$ with imperfect expansion.
(Recall that sparse vectors, which are qualitatively similar to indicator vectors, have larger $4$-norm than typical unit vectors.)
One catch is that $v$ must be completely contained in the span of eigenvectors of $G$ of magnitude at least $1/2$, which is a stronger requirement than $\iprod{v, Gv} \geq 1/2$.

\begin{theorem}[Consequence of Theorem 2.4 in \cite{DBLP:conf/stoc/BarakBHKSZ12}]
If there is $v$ in the span of eigenvectors of $G$ with eigenvalue at least $1/2$ such that $\|v\|_4^4 \geq 1/\delta n$, then $G$ contains a set $S$ of $\delta n$ vertices having expansion $\Phi_G(S) \leq 1-c$ for a universal $c > 0$.
Furthermore, $S$ may be found in polynomial time from $G$ and $v$.
\end{theorem}

While the vertex set $S$ from this result can be found in polynomial time, Barak et al.'s procedure to find $S$ from $v$ is complex.
In particular, it departs from the elegance of Cheeger's inequality that $S$ can be taken to be a level set of $v$.
Our tools give a simple proof of the following theorem, which we believe is novel -- it directly characterizes the small set which can be recovered from $v$ with large $4$-norm in terms of level sets of $v$ and the random walk on $G$.
\begin{theorem}\label{thm:2-to-4-rounding}
If there is $v$ in the span of eigenvectors of $G$ with eigenvalue at least $1/2$ such that $\|v\|_4^4 \geq 1/\delta n$, then there is a level set $S$ of the vector $w$ defined by $w_i = |v_i|$ which has the following property.
For some $t \leq O(\log n)$ there is level set of $G^t 1_S + G^{t+1} 1_S$ of size at most $O(\delta \log(1/\delta) \cdot n)$ having expansion $\Phi_G(S) \leq 1-c$ for a universal $c > 0$.
Here, $1_S$ is the $0/1$ indicator vector for the set $S$.
\end{theorem}
Qualitatively, our theorem says that an analytically-sparse $v$ in the high eigenspaces of $G$ has a level set $S$ such that if the random walk on $G$ is initialized to the uniform distribution on $S$, eventually the random walk ``discovers'' a small cut of imperfect expansion.
Thus, at the cost of a factor $\log(1/\delta)$ in the size of $S$ as compared to the result of Barak et al., we recover some of the elegance of Cheeger's rounding procedure for turning $v$ into a cut.
We describe the proof of Theorem~\ref{thm:2-to-4-rounding} in Appendix~\ref{sec:sparse-round}.

\section{Preliminaries}
\label{sec:prelims}

\subsection{Spectral graph theory}

Let $G = (V,E)$ be an $n$-node graph.
We also denote by $G$ the stochastic $n \times n$ random walk matrix associated to the graph $G$.

\begin{definition}[Isotropic spectral embedding]
  Let $\Pi_{1/2} \in\R^{n \times n}$ be the projector to the span of eigenvectors of $G$ with eigenvalues at least $1/2$.
  Let $A$ be a matrix such that $AA^\top = \Pi_{1/2}$.
  Without loss of generality, take the first column of $A$ to be $\mathbf{1}/\sqrt n$, the (scaled) all-$1$s vector.

  Let $a_1,\ldots,a_n$ be the rows of $A$.
  We say that $(a_1,\ldots,a_n)$ is the spectral embedding of $G$, and if $b_i = \sqrt{n} a_i$ we say that $(b_1,\ldots,b_n)$ is the isotropic spectral embedding of $G$.

\end{definition}
\noindent
We will need the following basic facts; the proofs are elementary and omitted.

\begin{fact}[Mean of a spectral embedding]
  \label{fact:spectral-embed-mean}
  Let $G$ be a graph and let $\Pi_{1/2}$ be the projector to the span of eigenvectors of $G$ of eigenvalue at least $1/2$.
  Let $a_1,\ldots,a_n$ be the rows of the matrix $A$ where $AA^\top = \Pi_{1/2}$; without loss of generality assume the first column of $A$ is the vector $\frac 1 {\sqrt n} \cdot \mathbf{1}$.
  Then $\frac 1 n \sum a_i = (1/\sqrt{n},0,0,\ldots,0)$.
\end{fact}

\begin{fact}
\label{fact:spectral-embed-iso}
$\E_{i \sim [n]} b_i b_i^\top = I$.
\end{fact}

For any $S \subseteq V$, we denote by $\1_S \in \{0,1\}^n$ the $0/1$ indicator vector of $S$.
For $v,w \in \R^n$ we often employ the usual Euclidean inner product $\iprod{v,w} = \sum_{i \leq n} v_i w_i$.

If $S \subseteq V$ is a subset of vertices in $G$, its expansion is the probability that a random walk initialized inside $S$ leaves $S$ in one step: $\Phi_G(S) = 1 - \frac 1 {|S|} \cdot \iprod{\1_S, G \1_S}$.
We define the \emph{expansion profile} of a graph $S$: for every $\delta > 0$, let $\Phi_G(\delta) = \inf_{|S| = \delta n} \Phi_G(S)$.
We also let $\Phi_G^\leq (\delta) = \inf_{|S| \leq \delta n} \Phi_G (S)$ be a slightly modified version of expansion profile which takes into account all sets of size at most $\delta n$, rather than exactly $\delta n$.

A consequence of Lemma~\ref{lem:local-cheeger} is a local Cheeger inequality concerning the quadratic form $\iprod{f,G^2 f} = \|Gf\|^2$ rather than $\iprod{f, Gf}$.
The proof is standard -- see the appendix.\footnote{Theorem 2.1 in \cite{d-to-1-steurer} is identical to Lemma~\ref{lem:local-cheeger-2} but is stated with the conclusion $\Phi_G(S) \leq 1 - \Omega(\e^2)$ rather than $\Phi_G(S) \leq 1- \Omega(\e^4)$; however the only proof we are aware of appears to require the extra factor of $1/\e^2$. Generally $\e$ is taken to be a tiny constant, so the difference is just one of constant factors.}

\begin{lemma} \label{lem:local-cheeger-2}
  Let $G$ be an $n$-node regular graph, and let $\eps, \delta, \gamma$ be so that $0 < \delta \leq \gamma$ and $\eps > 0$.
  Let $f \in \R^n$ have nonnegative coordinates, and suppose that $\|G f\|^2 \geq \e \|f\|^2$ and $\|f\|^2 \geq \frac{\gamma \|f\|_1^2}{\delta n}$.
  There is a level set $S$ of the function $g = f + Gf$ with size at most $\delta n$ and expansion $\Phi_G(S) \leq 1- \Omega(\gamma \e^4)$.
\end{lemma}

We will also require the following slight modification to Lemma~\ref{lem:local-cheeger-2}, which states that in the special case of $f$ being an indicator function for a subset, then we may additionally assume that the level set with poor expansion is additionally not too small.
We are not aware of a black-box proof of Lemma~\ref{lem:local-cheeger-2.5} from Lemma~\ref{lem:local-cheeger}, but our proof is a modification of the proof of Lemma~\ref{lem:local-cheeger} found in \cite{steurer-thesis}.
For completeness we prove this lemma in the appendix.
\begin{lemma} \label{lem:local-cheeger-2.5}
  There exist universal constants $0 < c < C$ such that the following holds.
  For every $G$ an $n$-node regular graph and every small enough $\eps, \delta, \eta$,
  if $S \subset [n]$ has $|S| = \delta n$ and $f = \1_S$ has $\|G f\|^2 \geq \e \|f\|^2$,
  then there is a level set $T$ of the function $g = (1 - \eta) f + \eta Gf$ with size $|T| \in \Brac{c \eta \eps^2 \delta n, C \frac{\delta n}{\eta^2 \eps^2}}$
  and expansion $\Phi_G(T) \leq 1- \Omega(\eta^2 \e^2)$.
  Moreover, if there is $R \subseteq [n]$ with $\Phi_G(R) \leq \e/100$ and $|R| = \delta n$, and if $S \cap R = \emptyset$, then also $T$ exists satisfying the previous properties and having $T \cap R = \emptyset$.
\end{lemma}

\subsection{Small-Set Expansion Hypotheses}

Our reductions in this paper are from small-set expansion problems, which are conjectured to be computationally difficult to solve.
At a high level, these assumptions say that it is hard to verify whether or not there exists a small set in a graph which does not expand well into the rest of the graph.
There are two canonical versions of this Small-Set Expansion Hypothesis (SSEH) which the literature appears to consider interchangeable.
However, for us it will be important to distinguish between the two.
The first (and original) version of SSEH concerns $\Phi_G(\delta)$:
\begin{hypothesis}[$=$-Small-Set Expansion Hypothesis ($\SSEHeq$) \cite{DBLP:conf/stoc/RaghavendraS10}]
For every constant $\epsilon > 0$ there is a small-enough $\delta > 0$ such that the following problem is $\NP$-hard.
Given a graph $G$, distinguish between $\Phi_G(\delta) \geq 1-\epsilon$ and $\Phi_G(\delta) \leq \epsilon$.
\end{hypothesis}
\noindent
In particular, this statement is only about sets of size exactly $\delta n$.
The second version of SSEH is essentially identical, except using $\Phi_G^\leq(\delta)$ instead of $\Phi_G(\delta)$.
\begin{hypothesis}[$\leq$-Small-Set Expansion Hypothesis ($\SSEHleq$)]
For every constant $\epsilon > 0$ there is a small-enough $\delta > 0$ such that the following problem is $\NP$-hard.
Given a graph $G$, distinguish between $\Phi_G^\leq(\delta) \geq 1-\epsilon$ and $\Phi_G^\leq(\delta) \leq \epsilon$.
\end{hypothesis}

We are not aware of any equivalences or implications between these two (apparently very similar) problems. 
However, both versions of the problem have been widely used and called the ``Small-Set Expansion Hypothesis'' in the literature, see e.g.~\cite{DBLP:conf/stoc/BarakBHKSZ12}.

We remark that while these two problems are very similar, there do appear to be some subtle qualitative differences between them.
In particular, in the context of this paper, $\SSEHeq$ (and variants thereof) implies hardness for problems related to resilience, whereas $\SSEHleq$ implies hardness for problems related to bounded moments.
At a high level, this is because bounded moments is equivalent to resilience at every scale (see Corollary~\ref{cor:resilience-to-moments}), and thus to control moments, we need to know what occurs at all sets of size at most $\delta$, not just in a neighborhood around $\delta$.

\input{content/moments-and-mean-shifts}

%% file: content/moments-and-mean-shifts.tex
%!TEX root = ../main.tex

%% file: content/spectral-bounds.tex
%!TEX root = ../main.tex

\section{Conditional Means of Small Sets in the Spectral Embedding}
\label{sec:spectral-bounds}

In this section we prove the following two key lemmas, which characterize the spectral embeddings of small sets of vertices in small-set expanders.
They suggest the following perspective on embeddings of small-set expanders, which is at the heart of all our arguments: if $G$ is a $(\delta,\epsilon)$-small-set expander, small sets of vectors in its spectral embedding cannot have average too far from the origin, while a small nonexpanding set in $G$ embeds to a set of vectors whose average is far from the origin.

Slightly more formally, for every $S \subseteq [n]$ with $|S| \leq \delta n$, if $\Phi_G(S) \geq 1 - \e$, then 
\[
\Norm{\tfrac 1 {|S|} \sum_{i \in S} b_i } \approx \e^{\Omega(1)} / \sqrt{\delta}\mper
\]
(At least, if $S$ is not too small.)
On the other hand, if $\Phi_G(S) \leq \e$, then
\[
\Norm{\tfrac 1 {|S|} \sum_{i \in S} b_i } \approx 1/\sqrt{\delta} \gg \e^{\Omega(1)}/\sqrt{\delta}\mper
\]

Now we make this formal.
The first lemma shows that a small non-expanding set in a graph $G$ has a spectral embedding far from the origin.
It has been observed several times before (see e.g. \cite{DBLP:conf/stoc/BarakBHKSZ12}).
We include the proof in the Appendix for completeness.

\begin{lemma}
\label{lem:spectral-lb}
  Suppose $G$ is an $n$-node graph.
  Let $b_1,\ldots,b_n$ be the isotropic spectral embedding of $G$.
    Then, every $T \subseteq [n]$ satisfies
    \[
    \Norm{ \frac 1 {|T|} \sum_{i \in T } b_i }^2 \geq \frac n {|T|} \cdot \Paren{\frac 12 - \Phi_G(T)}\mper
    \]
\end{lemma}

The second lemma shows that if $G$ is a small-set expander then every small set of vectors in its spectral embedding has mean near the origin.
By correctly setting parameters, something qualitatively similar would follow as a corollary of Theorem 2.4 in \cite{DBLP:conf/stoc/BarakBHKSZ12}, but our proof is much simpler than that route.
We show that for such a set $T$, if $\Norm{\frac 1 {|T|} \sum_{i \in T} b_i}$ were too large, then eventually the random walk on $G$, initialized to the uniform distribution on $T$, would find a small set with small expansion.

\begin{lemma}
\label{lem:spectral-ub}
  Let $G$ be an $n$-node graph.
  Suppose $\e,\delta$ are such that $\Phi_G^{\leq}(\delta) \geq 1-\e$, and $\e < \e_0$ for some universal constant $\e_0 > 0$.
  Let $b_1,\ldots,b_n$ be the isotropic spectral embedding of $G$.
  For every $T \subseteq [n]$ with $|T| \leq \delta n$,
  \[
  \Norm{\frac 1 {|T|} \sum_{i \in T} b_i } \leq C' \exp \Paren{ C \cdot \frac{\log(\delta n/|T|)}{\log(1/\e)}} \cdot \frac {\e^{1/10}} {\sqrt{\delta}}\mcom
  \]
  where $C', C > 0$ are universal constants.
\end{lemma}
\begin{proof}
  Up to scaling, $\Norm{\frac 1 {|T|} \sum_{i \in T} b_i }$ is the magnitude of the projection of the uniform probability distribution $\1_T/ |T|$ on $T$ into the span of eigenvectors of $G$ with eigenvalue at least $1/2$.
  We first argue that this magnitude is not affected by too much if we replace $\1_T / |T|$ with $G^t \cdot \1_T / |T|$, which is the probability distribution which results from running the random walk in $G$ for $t$ steps.

  To see this, first express $\frac 1 {|T|} \sum_{i \in T} b_i$ in terms of the indicator vector $\1_T \in \{0,1\}^n$:
  \[
  \Norm{\frac 1 {|T|} \sum_{i \in T} b_i}^2 = \frac n {|T|^2} \cdot \Norm{A^\top \1_T}^2 = \frac n {|T|^2} \sum_{i \, : \, \lambda_i \geq 1/2} \iprod{v_i, \1_T}^2
  \]
  where the columns of $A$ are the eigenvectors $v_i$ of $G$ with eigenvalue $\lambda_i$ at least $1/2$.
  For any $t \in \N$, note that
  \[
  \Norm{ A^\top G^t \1_T}^2 = \sum_{i \, : \, \lambda_i \geq 1/2} \iprod{G^t v_i, \1_T}^2 = \sum_{i \, : \, \lambda_i \geq 1/2} \lambda_i^t \iprod{v_i,\1_T}^2 \geq 2^{-t} \Norm{A^\top \1_T}^2\mper
  \]

  Our aim is to use the local Cheeger inequality to control $\Norm{A^\top G^t \1_T}^2$, which will be possible so long as the collision probability of $G^t \1_T$ is like that of the uniform distribution on a set of size at most $\delta n$.
  First, since $\Pi_{1/2} \preceq I$, we have $\Norm{A^\top G^t \1_T}^2 \leq \Norm{G^{t} \1_T}^2$.

  By the local Cheeger inequality (Lemma~\ref{lem:local-cheeger-2}) with $\gamma = \e^{0.1}$, there is a constant $C$ such that for every $t$, either $\|G (G^{t} \1_T)\|^2 < C \e^{0.1} \|G^{t} \1_T\|^2$ or $\|G^{t} \1_T\|^2 < \e^{0.2} \|G^{t} \1_T\|_1^2  / \delta n = \eps^{0.1} |T|^2  / \delta n$.
  (Otherwise the assumption $\Phi_G(\delta) \geq 1-\epsilon$ is violated.)
  The last equality follows because $\|\1_T\|_1 = |T|$ and $G$ preserves $1$-norms and nonnegativity of nonegative vectors.

  Pick $t$ to be the smallest integer such that the second alternative holds; i.e. $\|G^{t} \1_T\|^2 < \eps^{0.2} |T|^2 / \delta n$.
  (Such $t$ must exist because for smaller $t$ and small enough $\e$ the norm $\|G^{t} \1_T\|^2$ strictly decreases in each step of the random walk.)
  Then putting together our previous bounds,
  \[
  \Norm{\frac 1 {|T|} \sum_{i \in T } b_i}^2 \leq \frac {\eps^{0.2} n}{|T|^2} \cdot 2^t \cdot \|G^{t} \1_T\|^2 \leq 2^t \cdot \frac {\eps^{0.2} }{\delta}\mper
  \]

  We just need to bound $t$, the smallest integer such that $\|G^{t} \1_T\|^2 < \eps^{0.2} |T|^2 / \delta n$.
  If $\e$ is small enough, for every $t' < t$ we know that $\|G^{t'} \1_T\|$ is decreasing; in particular $\|G^{t'+1} \1_T\| < C {\e}^{0.1} \|G^{t'} \1_T\|$.
  Since $\|1_T\|^2 = |T|$, the number $t$ just has to be large enough that $\eps |T| / \delta n \geq (C {\e}^{0.1})^{t}$, which rearranges to $t \geq \frac{\log (\delta n / |T|)}{\log (C / \eps^{0.1})}$.
  Putting it together, we find
  \[
  \Norm{ \frac 1 {|T|} \sum_{i \in T} b_i }^2 \leq \exp \Paren{ C_1 \cdot \frac {\log (\delta n / |T|)}{\log1/\e}} \cdot \frac \e {2 \delta} \leq C_2 \exp \Paren{ C_1 \cdot \frac {\log (\delta n / |T|)}{\log1/\e}} \frac {\e^{0.2}} \delta
  \]
  for some universal $C_1, C_2 \geq 0$.
\end{proof}

\paragraph{Proving our main theorems from Lemma~\ref{lem:spectral-lb} and Lemma~\ref{lem:spectral-ub}}
We briefly describe how all our main results can be obtained using the preceding two lemmas and related ideas.
To prove Theorem~\ref{thm:resilience-cert-intro} on hardness of checking resilience of a set of points in $\R^d$, we take the set of points to be the spectral embedding of a graph $G$.
Then if $G$ is a small-set expander, one may see that no tail event in the uniform distribution over the embedding -- that is, no small set of vectors -- can deviate far from the origin, by Lemma~\ref{lem:spectral-ub}.\footnote{In reality, we must use a version of Lemma~\ref{lem:spectral-ub} which applies to $\Phi_G$ rather than $\Phi_G^\leq$ and takes only one step of the walk -- this lemma is really just the local Cheeger inequality. See Section~\ref{sec:hard-cert}.}
On the other hand, if $G$ has a small non-expanding set, resilience is immediately violated by applying Lemma~\ref{lem:spectral-lb}.\footnote{Actually, this is true only if the set has size $\Omega(\delta n)$, rather than perhaps having size, say, $\sqrt n$. This is why to prove hardness of resilience we need to start $\SSEH_=$.}

To prove Theorem~\ref{thm:moments-intro}, we again take the vectors $a_1,\ldots,a_n$ in the theorem statement to be the embedding of a graph $G$.
Lemma~\ref{lem:spectral-ub} leads to tail bounds for the uniform distribution over these vectors, which can then be translated into upper bounds on the moments of the distribution by Fact~\ref{fact:means-to-moments} in the case that $G$ is a small-set expander.
On the other hand, if $G$ has a small non-expanding set then Lemma~\ref{lem:spectral-lb} can be leveraged to prove lower bounds on the $p$-th moments of the uniform distribution on $a_1,\ldots,a_n$ for $p > 2$.
We then combine this with an averaging argument to gain better control over even more moments of the distribution when the graph is a small-set expander, while arguing that this averaging does not decrease the $p$-th moment in the presence of a small non-expanding set.

The proof of Theorem~\ref{thm:unique-intro} is similar, with one key difficulty.
To arrive at the end of the reduction in the setting of robust mean estimation under resilience, there must be a set of adversarially corrupted points, but the remaining points must be resilient.
This is where we critically leverage our strengthened small-set expansion hypothesis.
We strengthen the hypothesis in the following way: we suppose that small-set expansion remains hard if in one case we are promised that $G$ contains one small set $S$ with $\Phi_G(S) \leq \e$ but for all other $T$ with $|T| = \delta n$ and $T \cap S = \emptyset$ it holds that $\Phi_G(T) \geq 1-\e$.
The resulting control over deviations of small sets \emph{in the embedding of $[n] \setminus S$}, via local Cheeger inequalities adapted to account for the presence of the set $S$, allows us to show that the embedding of $[n] \setminus S$ is resilient.

%% file: content/hardness-of-resilience.tex
%!TEX root = ../main.tex

\section{Hardness of Certifying Conditions for Robust Mean Estimation}
\label{sec:hard-cert}

In this section we show that it is SSE-hard to decide whether a set of points satisfy resilience or bounded moments beyond the $\sqrt{\delta}$ barrier.
In particular, in this regime improved certification algorithms would likely lead to improved polynomial-time error rates for robust mean estimation under bounded moment or resilience assumptions.

Throughout this section, given an instance $G$ of SSE, as in Section~\ref{sec:spectral-bounds}, we will let $\Pi_{1/2}$ be the projector to the span of eigenvectors of $G$ having eigenvalue at least $1/2$, and we let $b_1,\ldots,b_n$ be the isotropic spectral embedding of $G$.

\subsection{Consequences of SSE}
To prove our hardness from SSE, we will actually reduce from the following more quantitative problems, which are known to be polynomial-time equivalent to SSE.

\subsubsection{Gap SSE}
The first allows us to go from $\SSEHeq$ to assuming control over all sets of size in some constant size window around $\delta n$.
In particular, consider the following variant of SSEH:
\begin{hypothesis}[Gap $=$-Small-Set Expansion Hypothesis ($\GapSSEHeq$) \cite{DBLP:conf/coco/RaghavendraST12}]
For all small-enough $\eps > 0$ and $M \geq 1$, there exists a small-enough $\delta \leq 1 / M$ so that the following problem is $\NP$ hard.
Given a graph $G$ on $n$ vertices, distinguish between:
\begin{quote}
    \textbf{yes:} There exists a non-expanding set $S \subseteq [n]$ with $|S| = \delta n$ an $\Phi_G (S) \leq \eps$.\\
    \textbf{no:} All sets $S \subseteq V$ with $|S| \in \Brac{\frac{\delta n}{M}, M \delta n}$ have $\Phi_G (S) \geq 1 - \eps$.
    \end{quote}
\end{hypothesis}
\noindent
Then it is known that this problem is equivalent to $\SSEHeq$:
\begin{proposition}[\cite{DBLP:conf/coco/RaghavendraST12}]
$\SSEHeq$ holds if and only if $\GapSSEHeq$ holds.
\end{proposition}

\subsubsection{Quantitative SSE}
It has been shown that in SSE, quantitative relationships between the parameters $\eps, \delta$ may be taken.
Specifically,~\cite{DBLP:conf/coco/RaghavendraST12} shows:
\begin{proposition}[\cite{DBLP:conf/coco/RaghavendraST12}, Theorem 3.5]
\label{prop:quant-SSE}
For every sufficiently small $\delta,\e,\gamma > 0$ the following problem is $\NP$-hard assuming $\SSEHleq$:
Given a graph $G$, distinguish between:
\begin{quote}
\textbf{yes:} $\Phi_G^\leq (\delta) \leq \e$.\\
\textbf{no:} For all $\delta' \in [0,1]$ it holds that $\Phi_G^\leq (\delta') \geq 1 - (\delta')^{\Omega(\e)} - \gamma/\delta'$.
\end{quote}
\end{proposition}

\subsection{Hardness of certifying resilience}

In this section, we prove the following theorem.

\begin{theorem}\label{thm:resilience-hardness}
  Under $\SSEHeq$, for all sufficiently small constants $s > 0$ there exists $\delta(s) > 0$ such that given $S \subseteq \R^d$ it is $\NP$-hard to distinguish between:
  \begin{quote}
    \textbf{yes:} the uniform distribution $X$ on $S$ is $(s\sqrt{\delta},\delta)$-resilient. \\
    \textbf{no:} there is an event $A$ in the uniform distribution on $S$ such that $\Pr(A) = \delta$ and $\Norm{ \E X \, | \, A - \E X } > 0.4 \cdot \sqrt{\delta} \cdot \tfrac{1 -\Pr A}{\Pr A}$.
    \end{quote}
\end{theorem}

\begin{proof}[Proof of Theorem~\ref{thm:resilience-hardness}]
  Let $\eps > 0$ be sufficiently small.
  We start with an instance $G$ of $\GapSSEHeq$ with parameter $\eps$, $M \geq 1$ to be set later, and corresponding $\delta = \delta(\e,M)$.
  Our reduction is simple: we let the set $S$ be $S = \{b_i\}_{i = 1}^n$.
  Observe that an event $E$ in the uniform distribution supported $S$ directly corresponds to a subset $T \subseteq S$, and moreover $\Pr E = |T| / n$.
  We verify that an efficient algorithm certifying $(s\sqrt{\delta},\delta)$-resilience of $S$, for some $s = s(\e)$, would solve SSE.
  (We will show $s = \Theta(\e^{1/8})$ suffices.)

  There are two cases to check.
  Suppose there exists a set $T \subset [n]$ with $|T| = \delta n$, so that $\Phi_G (T) \leq \eps$.
  Let $A$ be the event associated to that set.
  Then by Lemma~\ref{lem:spectral-lb}, we have
\begin{align*}
  \Norm{\E X | A}^2 &= \Norm{\frac{1}{T} \sum_{i \in T} b_i}^2 \\
  &\geq \frac{1/2 - \eps}{\Pr A} \; ,
\end{align*}
and so in particular since $\Pr A = \delta$, we have $\Norm{\E X | A} \geq \sqrt{1/2 - \e} \cdot \sqrt{\delta} \cdot \frac{1}{\Pr A} \geq \sqrt{\delta} \cdot \frac{1}{2\Pr A}$, for $\e$ sufficiently small.
Since $\E X = e_1$ and therefore $\Norm{\E X} = 1 \ll \Norm{\E X| A}$ for $\delta$ small, in this case the resulting set $S$ is in the {\bf no} case for resilience.

  We now check the other case.
  Suppose $\Phi_G (S) \geq 1 - \eps$ for all $S$ with $|S| \in [\delta n / M, M \delta n]$.
  We wish to verify that in this case the resulting distribution is in the {\bf yes} case for resilience.
  First, observe that for any set $T \subseteq [n]$ with associated event $E$, we have the bound
  \begin{align*}
  \Norm{\frac{1}{|T|} \sum_{i \in T} b_i}^2 &= n \cdot \Norm{A^\top \frac{\1_T}{|T|}}^2 \stackrel{(a)}{\leq} \frac{n}{|T|} \leq \frac{1}{\Pr E} \; , 
  \end{align*}
  where (a) follows since $A$ has spectral norm at most $1$.
  (Here $A$ is the matrix such that $AA^\top = \Pi_{1/2}$.)
  Let $r$ be a constant to be optimized later.
 If $\Pr E < r \delta$, then immediately $\Norm{\E X | E} \leq \frac{\sqrt{r \delta}}{\Pr E}$.
 On the other hand, if $\Pr E \in [r \delta, \delta]$, then if $T$ is the associated set, we must have 
\begin{align*}
\Norm{ \frac 1 {|T|} \sum_{i \in T } b_i }^2 &= n \Norm{ A^\top \frac{\1_T}{|T|}}^2 \\
&\stackrel{(a)}{\leq} 4n \Norm{G \frac{\1_T}{|T|}}^2 \\
&\stackrel{(b)}{\leq} \frac{4n \sqrt{\e}}{r} \Norm{\frac{\1_T}{|T|}}^2 \\
&= \frac{4 \sqrt{\e}}{r \Pr E} \; ,
\end{align*}
where (a) follows since $A A^\top = \Pi_{1/2} \preceq 4GG^\top$, and (b) follows since if we let $M = O \Paren{\frac{1}{r^2 \eps^2}}$, then this follows from Lemma~\ref{lem:local-cheeger-2.5} with $\eta = r$ (as otherwise we would witness a set with size in $[\delta n / M, M \delta n]$ with $\Phi_G (S) < 1 - \eps$).
As a result, we have
\[
\Norm{\E X | E} \leq \frac{2 \eps^{1/4}}{r^{1/2} \sqrt{\Pr E}} \leq \frac{2 \eps^{1/4} \sqrt{\delta}}{r^{1/2} \Pr E} \; .
\]
Thus, if we let $r = \e^{1/4}$, then in all cases, we have
\[
\Norm{\E X | E} \leq \frac{2 \eps^{1/8} \sqrt{\delta}}{\Pr E} \; .
\]
Since again $\Norm{\E X} = 1$, this implies that for $\delta$ sufficiently small, $S$ is $(s \sqrt{\delta}, \delta)$-resilient, for $s = \Theta(\e^{1/8})$.
Thus we are in the {\bf yes} case for resilience.
Our choice of $s = \Theta(\e^{1/8})$ ensures that $\SSEHeq$ applies for all small-enough $\e$ and hence for all small-enough $s$.
This completes the proof of correctness of the reduction.
%  \begin{align*}
%    \Norm{ \frac 1 {|T|} \sum_{i \in T } b_i }^2 &\leq C' \Paren{\frac{\delta n}{|T|}}^{C / \log (1 / \e)} \cdot \frac{\sqrt{\eps}}{\sqrt{\delta}} \\
%    &\stackrel{(a)}{\leq} \frac{\delta}{\Pr E} \frac{C' \sqrt{\eps}}{\sqrt{\delta}} \leq \frac{s \sqrt{\delta}}{\Pr E} \; ,
%  \end{align*}
%  where $s = s(\e) =  C' \sqrt{\eps}$ satisfies $s < c$ for $\eps$ sufficiently small, where (a) follows by the definition of $E$, and if we choose $\eps$ sufficiently small so that $C / \log (1 / \eps) \leq 1$.
\end{proof}

\subsection{Hardness of certifying bounded moments}
This section is dedicated to the proof of the following theorem:
\begin{theorem}
\label{thm:moment-cert}
Under $\SSEHleq$, there exists a constant $s > 0$ such that for any $q >2$, $c > s$, $t \in (2, q]$, given $S \subseteq \R^d$ it is $\NP$-hard to distinguish the cases:
  \begin{quote}
    \textbf{yes:} the uniform distribution $X$ on $S$ satisfies
    \[
      \sup_{\Norm{v} = 1} \Abs{\iprod{v, X - \E X}}^r \leq (s r)^{r/2}\; ,
    \]
    for all $2 < r \leq q$.\\
    \textbf{no:} there exists a unit vector $v \in \R^d$ so that
   \[
      \sup_{\Norm{v} = 1} \Abs{\iprod{v, X - \E X}}^t > (c t)^{t / 2} \; .
    \]

    \end{quote}
\end{theorem}
\noindent
We first show that the following intermediate problem is $\NP$-hard under $\SSEHleq$:
\begin{lemma}
\label{lem:moment-cert1}
There exists a universal constant $c \in [0,1]$ such that for all $q > 2$ and all small-enough $\delta$ the following problem is $\NP$-hard assuming $\SSEHleq$.
Given a set $S$ of $n$ points in $\R^d$ so that the uniform distribution $X$ over $S$ is isotropic, distinguish between:
\begin{quote}
{\bf yes:} There exists an event $E$ with probability $\Pr E \leq \delta$ so that $\Norm{\E X | E - \E X} \geq \frac{0.4}{\sqrt{\Pr E}}$.
In particular, by Fact~\ref{fact:means-to-moments}, this implies $\E \Abs{\iprod{v, X}}^r > \frac{0.4^r}{\delta^{r/2 - 1}}$ for some unit vector $v$, and any $r > 2$.\\
{\bf no:}  $\E |\iprod{v, X}|^r \leq \frac{\delta^{cr}}{\delta^{r/2-1}}$ for all unit vectors $v$ and all $r \in (2, q]$.
\end{quote}
\end{lemma}

We remark that this lemma (with different terminology) is very similar to the reduction presented in~\cite{DBLP:journals/corr/abs-1205-4484}, in their proof that SSE implies hardness for certifying $2\to4$ norms of tensors.
We give a proof here which simplifies and generalizes several key steps in their argument, and which gives us stronger guarantees which will be useful later.

The proof requires some bookkeeping, but the approach is simple.
We will take $S$ to be the isotropic spectral embedding of a graph $G$.
The {\bf yes} case is easy to establish
For the {\bf no} case, we first observe (Fact~\ref{fact:means-to-moments}) that moment bounds of the type in Lemma~\ref{lem:moment-cert1} are essentially equivalent to large-deviation tail bounds -- i.e. inequalities of the form $\Pr(\iprod{X,v} > t) \leq p(t)$ for unit vectors $v$ and various deviation magnitudes $t$.
We obtain such deviation inequalities from Lemma~\ref{lem:spectral-ub}, which shows that no small set of vectors in the spectral embedding of a small-set expander can deviate far from the origin.

\begin{proof}[Proof of Lemma~\ref{lem:moment-cert1}]
Fix $q > 2$.
Let $G$ be an instance of the problem given in Proposition~\ref{prop:quant-SSE} on $n$ vertices with $\delta < 0.05$ and $\epsilon < 0.05$ sufficiently small that Proposition~\ref{prop:quant-SSE} applies, and $\gamma = \delta^{1+\e}$.
Let $b_i$ be the isotropic spectral embedding of $G$.
Let $S = \{b_i\}_{i = 1}^n$.
We now verify that this set achieves the desired properties.

First, suppose that there exists a set $T \subseteq [n]$ with $|T| \leq \delta$ so that $\Phi_G (T) \leq \e$.
Then, by Lemma~\ref{lem:spectral-lb}, we have
\[
\Norm{\E X | E}^2 \geq \frac{0.5 - \e}{\Pr E} \; ,
\]
and so since $\Norm{\E X}^2 = 1 \ll \frac{0.45}{\Pr E}$, we have $\Norm{\E X | E - \E X} \geq \frac{0.4}{\sqrt{\Pr E}}$, for $\delta < 0.05$.
Thus, in this case the set $S$ belongs to the {\bf yes} case of Proposition~\ref{prop:quant-SSE}.

On the other hand, suppose that $\Phi_G^{\leq} (\delta') \geq 1 - (\delta')^{\Omega(\e)} - \gamma/\delta'$ for all $\delta' \in [0,1]$.
Fix any $r \in (2, q]$.
Our goal will be to use Fact~\ref{fact:means-to-moments}, which for any $s > r$ supplies the following bound on $\E | \iprod{v,X} - \E \iprod{v,X}|^r$ for any unit $v$ (by elementary integration):
\[
\E | \iprod{v,X} - \E \iprod{v,X}|^r \leq \sup_{E} \, (2 \Pr E)^{r/s} \cdot |\E \iprod{X,v} | E - \E \iprod{X,v} |^r \cdot \frac s {s-r}\mcom
\]
where the supremum is over all events $E$.

By Cauchy-Schwarz, for any unit $v$ and event $E$,
\begin{align} \label{eq:moment-1}
|\E \iprod{v,X} | E - \E \iprod{v,X}| \leq \| \E X |E - \E X\|\mper
\end{align}
So,
\begin{align}\label{eq:moments-2}
\E | \iprod{v,X} - \E \iprod{v,X}|^r \leq \sup_{E} \, (2 \Pr E)^{r/s} \cdot \| \E X | E - \E X \|^r \cdot \frac s {s-r}\mper
\end{align}
We choose $s = r \cdot \tfrac{\log(1/\delta)}{\log(1/\delta) -1}$, so that $s/(s-r) = \log(1/\delta)$.

We will bound the supremum in \eqref{eq:moments-2} by separately considering two cases: $\Pr E \leq \delta/2$ and $\Pr E > \delta/2$.
First, let $E$ have $\Pr E \leq \delta/2$.
By our choice of $\gamma$, we know that
\[
\Phi_G(\delta) \geq 1 - \delta^{\Omega(\e)}\mper
\]
Using this in conjunction with Lemma~\ref{lem:spectral-ub}, we know that there exist universal constants $C,C' > 0$ so that
\[
\|\E X | E \| \leq C' \Paren{\frac {\delta}{\Pr(E)}}^{C/\e \log(1/\delta)} \cdot \frac{\delta^{\Omega(\e)}}{\sqrt{\delta}}
\]
and hence by triangle inequality
\[
\| \E X | E - \E X \| \leq C' \Paren{\frac {\delta}{\Pr(E)}}^{C/\e \log(1/\delta)} \cdot \frac{\delta^{\Omega(\e)}}{\sqrt{\delta}} + 1
\]
because $\|\E X \| = 1$.
For small-enough $\delta$, we have $C' \Paren{\frac {\delta}{\Pr(E)}}^{C/\e \log(1/\delta)} \cdot \frac{\delta^{\Omega(\e)}}{\sqrt{\delta}} \geq 1$ for all $E$, and hence
\[
\| \E X | E - \E X \| \leq 2 C' \Paren{\frac {\delta}{\Pr(E)}}^{C/\e \log(1/\delta)} \cdot \frac{\delta^{\Omega(\e)}}{\sqrt{\delta}}\mper
\]

Returning to the expression from \eqref{eq:moments-2},
\[
(2 \Pr E)^{r/s} \cdot \| \E X | E - \E X \|^r \cdot \frac s {s-r} \leq (2 \Pr E)^{r/s} \cdot \Paren{\frac{2C' \delta^{\Omega(\e)}}{\sqrt{\delta}}}^{r} \cdot \Paren{\frac {\delta}{\Pr(E)}}^{Cr/\e \log(1/\delta)} \cdot \log \frac 1 \delta \mper
\]
By elementary algebra, using our choice of $s$ and the bound $\Pr E \leq \delta /2$, so long as $\log(1/\delta) > Cr/\e + 1$, we have
\[
(2 \Pr E)^{r/s} \Paren{\frac {\delta}{\Pr(E)}}^{Cr/\e \log(1/\delta)} \leq \delta^{1-1/\log(1/\delta)} \leq O(\delta)\mper
\]
So all together we got
\[
(2 \Pr E)^{r/s} \cdot \| \E X | E - \E X \|^r \cdot \frac s {s-r} \leq C \cdot \frac{(2C' \delta)^{\Omega(\e)\cdot r}}{{\delta}^{r/2-1}} \cdot \log \frac 1 \delta \mper
\]
for some (different) universal constant $C$.
For some universal $c_1,c_2$ if we choose $\eta = c_1 \delta^{c_2 \e}$, then for every small-enough $\delta$,
\[
(2 \Pr E)^{r/s} \cdot \| \E X | E - \E X \|^r \cdot \frac s {s-r} \leq \delta \cdot \Paren{\frac \eta \delta}^{r/2}\mper
\]

We turn to the case of $\Pr E > \delta /2$.
By hypothesis, $\Phi_G^{\leq}(\Pr E) \geq 1 - (\Pr E)^{\Omega(\e)}$.
So by Lemma~\ref{lem:spectral-ub} applied with $\delta' = \Pr(E)$, we obtain that for some universal $C$,
\[
\|\E X | E - \E X\| \leq C \cdot \Paren{\frac{(\delta')^{\Omega(\e)}}{\delta'}}^{1/2} \; ,
\]
(where we used $\|\E X\| =1$ again).

Using this to bound \eqref{eq:moments-2} for events with $\Pr E > \delta/2$, and recalling our choice of $s$ above, we have
\begin{align*}
(\Pr E)^{r/s} \cdot \| \E X | E - \E X\|^r \cdot \frac s {s-r} & \leq \Pr(E)^{1-1/\log(1/\delta)} \cdot C^r \cdot \Paren{\frac{(\delta')^{\Omega(\e)}}{\delta'}}^{r/2} \cdot \log(1/\delta) \\
& \leq \frac{O(\delta)^{\Omega(\e r)}}{\delta^{r/2-1}} \\
& \leq \delta \cdot \Paren{\frac {\eta}{\delta}}^{r/2}
\end{align*}
for small-enough $\delta$ and the choice of $\eta$ above; the second simplification just uses $\Pr(E) \geq \delta /2$.

We conclude by \eqref{eq:moments-2} that for all $\delta < \delta_0(\e,r)$ it holds that
\[
\E |\iprod{X - \E X,v}|^r \leq \frac{\eta^{r/2}}{\delta^{r/2 - 1}}\mper
\]
Thus picking $\delta < \min_{r \leq q} {\delta_0(\e,r)}$, we conclude that the set of vectors $S$ is in the {\bf no} case.

Finally, the distribution over $S$ is as stated not isotropic, because the first coordinate of every vector is $1$.
Indeed, it is a standard fact that the distribution which is simply the uniform distribution over the vectors in $S$ with the first coordinated removed is mean zero and isotropic.
However, it is easily to check that the proof above goes through for $S$ projected off of the first coordinate.
Then the resulting distribution is indeed isotropic, and satisfies all the desired guarantees as in the Lemma.
This completes the proof.
\end{proof}
\noindent
The second lemma we need to prove Theorem~\ref{thm:moment-cert} is the following inequality for $p$-th moments of sums of independent random variables.
\begin{fact}[Rosenthal's Theorem, see e.g.~\cite{johnson1985best}]
\label{fact:rosenthal}
Let $p \geq 2$, and let $X_1, \ldots, X_n$ be independent with $\E X_i = 0$ and $\E |X_i|^p < \infty$ for all $i = 1, \ldots, n$.
Then
\[
\E \Abs{\sum X_i}^p \leq (C_1 p)^p \cdot \Paren{\sum \E \Brac{|X_i|^p} } + (C_2 p)^{p /2} \cdot \Paren{\sum_{i = 1}^n \E \Brac{X_i^2}}^{p / 2} \; ,
\]
for some universal constants $C_1, C_2$.
\end{fact}

\begin{proof}[Proof of Theorem~\ref{thm:moment-cert}]
Let $S$ be an instance of the problem in Lemma~\ref{lem:moment-cert1} with parameters $q$ and $\delta$.
We show how to construct a set $S'$ over $n^{O(1/\delta)}$ points in time $n^{O(1/\delta)}$ so that a {\bf yes} instance of the problem in Lemma~\ref{lem:moment-cert1} is mapped to a {\bf yes} instance of the problem in Theorem~\ref{thm:moment-cert}, and similarly for {\bf no} instances.
Composing this reduction with Lemma~\ref{lem:moment-cert1} immediately yields Theorem~\ref{thm:moment-cert}.

To achieve this, we simply let $S'$ be the set
\[
  S' = \left\{ \sqrt{\frac \delta  \alpha} \sum_{i_1,\ldots,i_{\alpha/\delta}} X_i \, :  \, i \in [|S|] \right\} \; ,
\]
or equivalently, the uniform distribution $D'$ over $S'$ is the the sum of $\alpha / \delta$ i.i.d. samples from the uniform distribution $D$ over $S$, scaled by $\sqrt{\delta / \alpha}$.
Here $\alpha \leq 1$ is a parameter depending only on $q$ and $c$ to be tuned later.
Clearly $|S'| = n^{O(1 / \delta)}$ and can be constructed in time $n^{O(1 / \delta)}$ given a construction of $S$.
We now check soundness and completeness.

Suppose $S$ is an {\bf yes} instance from Lemma~\ref{lem:moment-cert1}.
Then there exists an event $E$ of $D$ with $\Pr E \leq \delta$ and a unit vector $v$ so that $\Abs{\E \iprod{v, X} | E - \E \iprod{v, X}} \geq \frac{0.4}{\sqrt{\Pr E}}$.
The event $E$ corresponds to some $T \subset S$ with $|S| \leq \delta |S|$.
Let $E'$ be the event in $D'$ that at least one $X_i$ in the sum belongs to $S$.
By standard estimates, $\Pr \Brac{X' \in S'} = 1 - (1 - \Pr E)^{\alpha / \delta} = \Omega(\alpha \Pr E / \delta)$ for $\Pr E \leq \delta$.
Moreover, since $\E_{X \sim D} X = 0$, we have that $\| \E_{X' \sim D'} X' | E' \| = \sqrt{\delta / \alpha} \cdot \| \E_{X \sim D} X | E \| \geq 0.4 \cdot \sqrt{\frac{\delta}{\alpha \Pr E}}$.
Hence, using the contribution of the event $E$ to the $t$-th moments to lower-bound them (Fact~\ref{fact:means-to-moments}), for $t \in (2, q]$, there exists some unit vector $v$ so that
\begin{align*}
\E \Abs{\iprod{v, X}}^{t} &\geq \Paren{\Pr E'} \cdot (0.4)^{t} \cdot \Paren{\frac{\delta}{\alpha \Pr E}}^{t / 2} \\
&\geq \Omega \Paren{\frac{\alpha \Pr E}{\delta}} (0.4)^{t} \cdot \Paren{\frac{\delta}{\alpha \Pr E}}^{t/2} \\
&\geq \Paren{\frac{1}{\alpha}}^{\Omega (t)} \Paren{\frac{\delta}{\Pr E}}^{t/2 - 1} \\
&\geq \Paren{\frac{1}{\alpha}}^{\Omega (t)} \geq (c t)^{t / 2} \; ,
\end{align*}
for $\alpha$ chosen such that $c = \frac{1}{t} (1 / \alpha)^{\Omega (1)}$.
Hence $S'$ is an instance of the {\bf yes} case.
%In particular, if $\alpha = \Omega (q)$, we have that $c = \omega(1)$.

Suppose on the other hand that $S$ is a {\bf no} instance.
Let $v$ be an arbitrary unit vector.
Let $X' \sim D'$, so that $X' = \sqrt{\delta / \alpha} \Paren{\sum_{i = 1}^{\alpha / \delta} X_i}$ where $X_i \sim D$ are independent.
Then, by Rosenthal's inequality (Fact~\ref{fact:rosenthal}) applied to the random variables $Z_i = \sqrt{\delta} \iprod{v, X_i}$, we see that there are universal constants $C_1,C_2$ such that for any $r \in (2, q]$,
\begin{align*}
\E \Abs{\iprod{v, X'}}^{r} &\leq (C_1 r)^r \cdot \Paren{(\delta / \alpha)^{r/2 - 1} \E_{X \sim E} |\iprod{v, X}|^r} + (C_2 r)^{r/2} \\
&\leq (C_1 r)^r \cdot \delta^{\Omega(r)} \cdot \Paren{1/\alpha}^{r/2 - 1} + (C_2 r)^{r/2} \; ,
\end{align*}
by Lemma~\ref{lem:moment-cert1}.
Using our previous choice for $\alpha$, we see that if $\delta$ is small enough as a funcion of $c,q$ then the second term dominates, and we get
\[
\sup_{\norm{v} = 1} \E \Abs{\iprod{v, X'}}^{r} \leq (s r)^{r/2} \; ,
\]
for some universal constant $s = O(1)$.
Thus in this case we are in the {\bf no} case.
This completes the proof.
\end{proof}

%% file: content/unique-sse.tex
%!TEX root = ../main.tex

\section{Unique-SSE and Robust Estimation}
\label{sec:unique-sse}

In this section we prove Theorem~\ref{thm:usseh-robust-mean} on hardness of robust estimation under \USSEH.

\begin{definition}[Almost-SSE]
  Suppose $G$ is an $n$-node graph.
  We say that $G$ is an almost $(\e,\delta)$ small set expander if:
  \begin{itemize}
    \item there is $S \subseteq [n]$ with $|S| = \delta n$ and $\Phi_G(S) \leq \e$, and
    \item every $T \subseteq [n]$ with $|T| = \delta n$ and $T \cap S = \emptyset$ has $\Phi_G(T) \geq 1-\e$\mper
  \end{itemize}
\end{definition}

\begin{hypothesis}[Unique Small-Set Expansion Hypothesis $\USSEH$]
  For every $\e > 0$ there is a small-enough $\delta > 0$ such that the following problem is $\NP$-hard.
  Given an $n$-node graph $G$, distinguish between the cases: \textbf{yes:} $G$ is an almost $(\e,\delta)$ small set expander, and \textbf{no:} $\Phi_G(\delta) \geq 1-\epsilon$.
\end{hypothesis}

\begin{problem}[$\alpha,\beta$-approximate robust mean estimation under resilience]\label{prob:resil-estimate}
  \textbf{Input:} $b_1,\ldots,b_n \in \R^d$ and $\delta > 0$, such that there exists $S \subseteq [n]$ with $|S| = (1-\delta)n$ which is $(\alpha \sqrt{\delta},\delta)$-resilient.
  \textbf{Output:} A vector $\hat{\mu} \in \R^n$ such that $\|\hat{\mu} - \E_{i \sim S} b_i \| \leq \beta \sqrt{\delta}$.
\end{problem}

\begin{theorem}\label{thm:usseh-robust-mean}
  Suppose \USSEH.
  There is an absolute constant $\beta^* < 1$ such that if for any constant $\alpha < \beta^*$ Problem~\ref{prob:resil-estimate} has a polynomial-time algorithm then $\Pclass = \NP$.
\end{theorem}

Our main tool is the ``moreover'' clause in Lemma~\ref{lem:local-cheeger-2.5} which allows for $G$ to be an almost $(\e,\delta)$ small set expander rather than a small set expander.
This allows us to prove the following result characterizing the means of embeddings of small sets in $G$ which do not overlap with the small non-expanding set.

\begin{lemma}\label{lem:unique-mean-shift}
  Suppose that $G$ is an $n$-node almost $(\e,\delta)$ small set expander for $\e < \e_0$, where $\e_0 > 0$ is a universal constant.
  Let $T \subseteq [n]$ have $|T| \leq \delta n$ and no intersection with the small non-expanding set in $G$.
  Let $b_1,\ldots,b_n$ be the isotropic spectral embedding of $G$.
  Then
  \[
  \Norm{\frac 1 {|T|} \sum_{i \in T} b_i } \cdot \frac {|T|} n \leq 2 \e^{0.05} \sqrt \delta\mper
  \]
\end{lemma}

\begin{proof}
  We proceed as in the proof of Theorem~\ref{thm:resilience-hardness}.
  By definition, $\frac 1 {|T|} \sum_{i \in T} b_i = A^\top \cdot \frac {\sqrt n} {|T|} 1_T$ where $A$ has columns which are the eigenvectors of $G$ with eigenvalue at least $1/2$.
  We will combine two bounds, one for $|T| \ll \delta n$ and one for $|T| \approx \delta n$.

  Firstly, because $\|A\| \leq 1$, we have
  \[
  \Norm{\frac 1 {|T|} \sum_{i \in T} b_i}^2 = n \cdot \Norm{A^\top \frac{1_T}{|T|}}^2 \leq n \cdot \Norm{1_T / |T|}^2 = \frac n {|T|}\mper
  \]
  Let $r = r(\e,\delta)$ be a constant to be chosen later.
  If $|T| \leq r \delta n$, then we find 
  \[
    \Norm{\frac 1 {|T|} \sum_{i \in T} b_i } \cdot \frac{|T|}{n} \leq \sqrt{\frac{|T|} n} \leq \sqrt{r \delta}\mper
  \]

  Now we address sets with sizes in the range $|T| \in [ r \delta n, \delta n]$.
  Here we will use a local Cheeger inequality -- Lemma~\ref{lem:local-cheeger-2.5}.
  We are interested in
  \[
  \frac n {|T|^2} \iprod{1_T, \Pi_{1/2} 1_T} \leq \frac{4 n}{|T|^2} \|G 1_T\|^2\mper
  \]
  Picking $\eta = \e^{0.1}$,
  if $\|G 1_T\|^2 \geq \e^{0.1} \|1_T\|^2$ then there is a set $R$ of size in the range $|R| \in [c \e^{0.3} r \delta n, C \delta n / \e^{0.4}]$ for some universal constants $c,C$, with expansion $\Phi_G(R) \leq 1 - \Omega(\e^{0.4})$.
  Furthermore, $R \cap S = \emptyset$.

  By subsampling at random or adding vertices as necessary, we find that there is a set $R'$ of size $\delta n$ which does not overlap $S$ and has expansion $\Phi_G(R') \leq 1 - \Omega(r \e^{0.8})$.
  Choosing $r = \e^{0.1}$ and $\e$ sufficiently small, this violates that $G$ is an almost $(\e,\delta)$ small set expander.
  So it must be that $\|G 1_T\|^2 \leq \e^{0.1}\|1_T\|^2 \leq \e^{0.1} \delta n$.
  We therefore find that for $|T| \in [r \delta n, \delta n]$,
  \[
    \Norm{\frac 1 {|T|} \sum_{i \in T} b_i} \cdot \frac {|T|} n \leq \frac{|T|}{n} \cdot \frac{2 \sqrt n}{|T|} \cdot \e^{0.05} \sqrt{\delta n} = 2 \sqrt{\e^{0.1} \delta}\mper
  \]
\end{proof}

\begin{proof}[Proof of Theorem~\ref{thm:usseh-robust-mean}]
  We will analyze the following reduction from small-set expansion to robust mean estimation under resilience.
  Let $\beta^*$ be a small-enough absolute constant.
  (We can choose it later).
  Let $\alpha < \beta^*$.

  Given an $n$-node graph $G$ and parameters $\e, \delta > 0$, let $b_1,\ldots,b_n$ be the isotropic spectral embedding of $G$.
  Let $\mu$ be the output of an oracle for Problem~\ref{prob:resil-estimate} with parameters $\alpha, \beta^*, \delta/2$ on input $b_1,\ldots,b_n$.
  Let $e_1 \in \R^d = (1,0,0,\ldots,0)$ be the first standard basis vector.
  If $\| \mu - e_1\| > 2 \beta^* \sqrt{\delta}$ then output \textbf{yes}.
  Otherwise output \textbf{no}.

  We need to show that there exists $\e > 0$ such that for all $\delta > 0$ the following two statements hold:

  \textbf{Soundness: } If $\Phi_G(\delta) > 1-\e$ then $\|\mu - e_1\| \leq 2 \beta^* \sqrt{\delta}$.

  \textbf{Completeness: } If $G$ is an almost $(\e,\delta)$ small set expander then $\|\mu - e_1\| > 2 \beta^* \sqrt{\delta}$.

  We address the statements in turn, beginning with soundness.
  By the proof of Theorem~\ref{thm:resilience-hardness}, if $\Phi_G(\delta) > 1-\e$ then the uniform distribution on $\{b_1,\ldots,b_n\}$ is $(2\e^{1/8} \sqrt{\delta}, \delta)$-resilient.

  Hence every subset of $S$ of size $(1-\delta/2)n$ is also $(4 \e^{1/8} \sqrt{\delta},\delta/2)$-resilient.
  Fix one such subset $S$.
  By Fact~\ref{fact:spectral-embed-mean}, we have $\E_{i \sim [n]} b_i = e_1$.
  Hence by resilience, $\| \E_{i \sim S} b_i - e_1 \| \leq 2 \e^{1/8} \sqrt{\delta}$.
  By the guarantee of our robust mean estimation oracle, so long as $2 \e^{1/8} \sqrt{\delta} \leq \alpha$ then $\|\mu - \E_{i \sim S} b_i\| \leq \beta^* \sqrt{\delta}$.
  By triangle inequality,
  \[
  \| \mu - e_1 \| \leq \|\mu - \E_{i \sim S} b_i \| + \|\E_{i\sim S} b_i - e_1 \| \leq (\beta^* + 2 \e^{1/8}) \sqrt{\delta} \leq 2 \beta^* \sqrt{\delta}
  \]
  for small-enough $\epsilon = \epsilon(\alpha,\beta^*)$.

  Now we move on to completeness.
  Let $S \subseteq [n]$ be the $\delta n$-size subset of vertices with $\Phi_G(S) \leq \e$.
  Let $v = \E_{i \sim S} b_i$ and let $w = \E_{i \notin S} b_i$.
  Since (by Fact~\ref{fact:spectral-embed-mean}) we have $\E_{i \sim [n]} b_i = e_1$, simple calculations show that
  \[
  w = \frac{e_1 - \delta v}{1-\delta}\mper
  \]
  This rearranges to
  \[
  e_1 - w = \frac {\delta v + \delta e_1 }{1-\delta}\mper
  \]

  We first establish that the set $\{b_i\}_{i \notin S}$ is $(4 \e^{0.05} \sqrt{\delta},\delta/4)$-resilient.
  Let $R \subseteq [n] \setminus S$ have size at most $|R| \leq \delta n /2$.
  Then by Lemma~\ref{lem:unique-mean-shift}, we have
  \[
  \Norm{\frac 1 {|R|} \sum_{i \in R} b_i} \leq 2 \e^{0.05} \sqrt{\delta} \cdot \frac n {|R|}\mper
  \]
  Hence by triangle inequality we have
  \[
  \Norm{\frac 1 {|R|} \sum_{i \in R} b_i - w} \leq 2 \e^{0.05} \sqrt{\delta} \cdot \frac {n}{|R|} + \|w\|
  \]
  and so finally
  \[
  \Norm{\frac 1 {|R|} \sum_{i \in R} b_i - w} \cdot \frac {|R|}{n} \leq 2 \e^{0.05} \sqrt{\delta} + \delta \cdot \|w\|\mper
  \]
  It follows that $\{b_i\}_{i \notin S}$ is $(2 \e^{0.05} + \delta \|w\|, \delta/4)$-resilient.
  By Lemma~\ref{lem:unique-mean-shift}, $\|w\| \leq 2 \e^{0.05} \sqrt{1/\delta}$.
  So ultimately, $\{b_i\}_{i \notin S}$ is $(4 \e^{0.05} \sqrt{\delta}, \delta/4)$-resilient.

  Therefore, we must have that $\|\mu - w\| \leq O(\e^{0.05} \sqrt{\delta})$.
  At the same time, by Lemma~\ref{lem:spectral-lb}, we have $\|v\|^2 \geq 1/2\delta$, so $\|w -e_1\| \geq \Omega(\sqrt{\delta})$.
  So,
  \[
  \| \mu - e_1\| = \| (\mu - w) + (w - e_1)\| \geq \|w - e_1\| - \|\mu - w\| \geq \Omega(\sqrt{\delta}) - O(\e^{0.05} \sqrt{\delta})\mper
  \]
  So, for sufficiently small $\beta^*$ and $\e$, we find that for all $\delta$, $\|\mu - e_1\| > 2 \beta^* \sqrt{\delta}$.
\end{proof}

%% file: content/open-problems.tex
%!TEX root = ../main.tex

\section{Conclusion and Open Problems}
\label{sec:open-probs}

In this paper we give evidence from worst case complexity assumptions that improving existing algorithms for robust mean estimation may be hard.
These results are far from complete, however, and there are a number of very interesting open questions in this area.

The most natural question is whether or not we can show that improving current algorithms for robust mean estimation assuming bounded moments or resilience is impossible under $\SSEH$.
There are a number of interesting sub-questions:
\begin{itemize}
\item
Can the uniqueness assumption be removed in the proof that $\USSEH$ implies improved robust mean estimation under resilience is $\NP$-hard? As far as we are aware it could even be that $\SSEH$ and $\USSEH$ are equivalent -- are they?.
\item
Does $\SSEH$ or a variant (such as $\USSEH$) imply that improved robust mean estimation is hard under bounded moment assumptions?
Our current techniques are unable to prove this for $\USSEH$: they would require an analogue of Lemma~\ref{lem:spectral-ub} in the setting that $G$ is a graph which contains a unique small non-expanding set.
That lemma requires running a random walk on the graph $G$ for about $\log n$ steps; we do not know how to ensure that such a random walk avoids entering the small non-expanding set (or, if it does, how to control its behavior across the nonexpanding cut).
\end{itemize}

Another interesting question is whether or not these techniques can be used to show hardness for other questions in robust estimation, such as list learning~\cite{DBLP:conf/stoc/CharikarSV17}, or robust sparse mean estimation~\cite{balakrishnan2017computationally}.
(Unlike for the main problems addressed in this paper, SQ lower bounds for these are already known \cite{diakonikolas2017statistical,diakonikolas2018list}.)
We conjecture that the current spectral-based algorithms for these problems are optimal, even with additional assumptions on resilience or moments.

It is also interesting to ask whether $\SSEH$-type assumptions can be avoided all together.
In addition to showing that approximating the $2 \rightarrow 4$-norm is $\SSEH$-hard, the authors of \cite{DBLP:conf/stoc/BarakBHKSZ12} also show it is $\NP$-hard assuming the Exponential Time Hypothesis.
That proof does not appear to easily adapt to our setting, however, because it is not clear the instance of the $2\rightarrow 4$-norm problem it produces can be transformed into a distribution with sub-Gaussian moments as we require, nor can we easily control the kind of tail events we require to prove hardness of resilience.
Nonetheless, it seems plausible that hardness for some robust estimation problem could be shown under assumptions weaker than $\SSEH$.

%% file: content/appendix.tex
%!TEX root = ../main.tex

\section{Omitted Proofs from Section~\ref{sec:prelims}}

\input{content/local-cheeger}

\subsection{Equivalence of Moments and Mean Shifts}
\label{sec:means-to-moments}

We will repeatedly use the following elementary fact, which proves a near equivalence of moment bounds and mean shifts for $\R$-valued random variables.

\begin{fact}
\label{fact:means-to-moments}
Let $X$ be a $\R$-valued random variable with mean zero, and let $q \geq 1$.
Then:
\begin{itemize}
\item {\bf Moment bounds implies bounded deviation} Suppose $\E |X|^q$ is finite.
Then for any event $A$, we have
$\Abs{ \E X \, | A } \leq \Paren{\frac{\E |X|^q}{\Pr [A]}}^{1/q}$.
\item {\bf Bounded deviation implies moment bounds} For any $p$, let
$C_{p} = \sup_A \Pr [A] \cdot \Abs{ \E X \, | A }^{p}$.
For every $p > q$,
$\E |X|^q \leq (2 C_p)^{q/p} \cdot \frac p {p-q}$.
\end{itemize}
\end{fact}
\noindent

\begin{proof}[Proof of Fact~\ref{fact:means-to-moments}]
We first prove the first implication.
By Holder's inequality, we have
\[
\Abs{ \E X 1_A } &\leq \Paren{\E \Abs{X}^q}^{1/q} \Pr [A]^{1 - 1/q} \; ,
\]
and so
\[
\Abs{\E X \, | A} = \frac{1}{\Pr [A]} \left| \E X 1_A \right| \leq \Paren{\frac{\E \Abs{X}^q}{\Pr [A]}}^{1/q} \; ,
\]
as claimed.

We now turn to the second implication.
For any $t \geq 0$,
\[
\Pr[|X| \geq t] = \Pr[X \geq t] + \Pr[X \leq -t] \leq \frac{C_{p}}{\Abs{\E X \, | \, X \geq t}^{p}} + \frac{C_{p}}{\Abs{ \E X \, | \, X \leq -t}^{p}} \leq \frac{2 C_{p}}{t^{p}}\mper
\]

Recall that $\E |X|^q = \int_0^\infty \Pr[ |X|^q \geq s ] \, ds$.
We will split this integral into two parts, because we know two different bounds on $\Pr [ |X|^q \geq s]$.
First of all, for any $s$ we have $\Pr[|X|^q \geq s] \leq 1$
Second of all, when $s > (2C_p)^{q/p}$ a better bound is given by $\Pr[|X|^q \geq s] \leq 2 C_p / s^{p/q} < 1$.
So,
\[
\E |X|^q = \int_0^\infty \Pr[ |X|^q \geq s] \, ds \leq \int_0^{(2C_p)^{q/p}} 1 \, ds + \int_{(2C_p)^{q/p}}^\infty \frac{2 C_p}{s^{p/q}} \, ds\mper
\]
The first integral is just $(2C_p')^{q/p}$.
The second is
\[
  \int_{(2C_p')^{q/p}}^\infty \frac{2 C_p'}{s^{p/q}} \, ds = \frac 1 {\tfrac pq - 1} \cdot [(2C_p)^{q/p}]^{-p/q + 1}
\]
so long as $p > q$.
(Otherwise the integral does not exist.)

Putting these together,
\[
\E |X|^q \leq (2C_p)^{q/p} + \frac 1 {\tfrac pq - 1} \cdot 2 C_p \cdot [(2C_p)^{q/p}]^{-p/q + 1} = (2C_p)^{q/p} \cdot \Paren{1 + \frac 1 {(\tfrac pq - 1)}}\mper
\]
Finally, note that $1 + 1/(\tfrac pq-1) = \tfrac{p-q}{p-q} + \tfrac{q}{p-q} = \tfrac p {p-q}$, which finishes the proof.
\end{proof}

\noindent
As a simple corollary of this, we observe that moment bounds are equivalent to resilience ``at every scale''.
For simplicity of exposition, we will state and prove the claim for $\ell_2$ norm, however, the claim holds much more generally as well.
This gives a novel characterization of resilience which may be of independent interest.
\begin{corollary}
\label{cor:resilience-to-moments}
Let $X$ be an $\R^d$-valued random variable with mean $\E X = \mu$, and let $q \geq 1$.
Then:
\begin{itemize}
\item {\bf Moment bounds imply multi-scale resilience} Suppose there exists a constant $C > 0$ so that $\E \iprod{v, X}^q \leq C$ for all unit vectors $v$.
Then, $X$ is  $(\frac{2 C^{1/q}}{\delta^{1/q - 1}},\delta)$-resilient for all $\delta \leq 1/2$.
\item {\bf Multi-scale resilience implies moment bounds} Let $p > q$, and let $C_p$ be so that $X$ is  
$(\frac{C_p^{1/p}}{\delta^{1/p - 1}},\delta)$-resilient for all $\delta \leq 1/2$.
Then, 
\[
\E \Abs{\iprod{v, X - \mu}}^q \leq (2 C_p)^{q / p} \frac{p}{p - q}
\]
for all unit vectors $v \in \R^d$.
\end{itemize}
\end{corollary}
\begin{proof}
We first prove the first implication.
Let $v$ be an arbitrary unit vector, let $\delta \in (0, 1/2)$ and let $A$ be an event with $\Pr [A] \leq \delta$.
Then, by our assumption and Fact~\ref{fact:means-to-moments}, we know that
\[
\Abs{\E \iprod{v, X} | A  - \iprod{v, \mu}  } \leq \Paren{\frac{C}{\Pr A}}^{1/q} = \frac{C^{1/q}}{\Pr[A]^{1/q - 1}} \cdot \frac{1}{\Pr A} \leq \frac{C^{1/q}}{\delta^{1/q - 1}} \cdot \frac{1}{\Pr A} \leq \frac{2 C^{1/q}}{\delta^{1/q - 1}} \frac{1 - \Pr A}{\Pr A} \; .
\]
Taking a supremum of this inequality over all unit vectors $v$ immediately yields the desired bound.

We now prove the other direction.
For any unit vector $v \in \R^d$, and any event $A$ with $\Pr A \leq 1/2$, and by our assumption of resilience (taking $\delta = \Pr A$), we have
\begin{equation}
\label{eq:resilience-to-moments}
\Pr[A] \cdot \Abs{\E \iprod{v, X} | A - \iprod{v, \mu}}^p \leq \Pr[A] \cdot \Norm{\E X | A - \mu}^p \leq C_p \; .
\end{equation}
Moreover, for any event $A$ with $\Pr A > 1/2$, we also have
\[
\Pr[A] \cdot \Abs{\E \iprod{v, X} | A - \iprod{v, \mu}}^p = \Pr[A] \cdot \Paren{\frac{\Pr A^c}{\Pr A}}^{p - 1} \Pr[A^c] \cdot \Abs{\E \iprod{v, X} | A^c - \iprod{v, \mu}}^p \leq C_p \; ,
\]
where the last inequality follows from~\eqref{eq:resilience-to-moments}.
Thus, by Fact~\ref{fact:means-to-moments}, we have
\[
\E \Abs{\iprod{v, X - \mu}}^q \leq (2 C_p)^{q / p} \frac{p}{p - q} \; ,
\]
as claimed.
\end{proof}
\noindent
We briefly remark that to generalize this statement to more general norms $\| \cdot \|$, it suffices to have the moment bound be taken over all unit vectors over the dual norm $\| \cdot \|^*$.
The proof is a fairly standard generalization of this argument and we omit the proof for simplicity of exposition.

\section{Omitted Proofs from Section~\ref{sec:spectral-bounds}}
\begin{proof}[Proof of Lemma~\ref{lem:spectral-lb}]
  Let $B = \sqrt{n} A$.
  We start by expanding:
  \[
  \Norm{ \frac 1 {|T|} \sum_{i \in T} b_i}^2 = \frac 1 {|T|^2} \cdot \1_T^\top BB^\top \1_T^\top = \frac n {|T|^2} \1_T^\top \Pi_{1/2} \1_T\mper
  \]
  Let $v_1,\ldots,v_n$ be the eigenvectors of $G$, with associated eigenvalues $\lambda_1,\ldots,\lambda_n$.
  Since $G$ is stochastic, $|\lambda_i| \leq 1$.
  So for any vector $v$ we have
  \[
    v^\top \Pi_{1/2} v = \sum_{i \, : \, \lambda_i \geq 1/2} \iprod{v,v_i}^2 \geq \sum_{i=1}^n \lambda_i \iprod{v,v_i}^2 - \frac 12 \cdot \|v\|^2 = v^\top G v - \frac 12 \cdot \|v\|^2 \mper
  \]
  Putting it together,
  \[
  \Norm{ \frac 1 {|T|} \sum_{i \in T} b_i}^2 \geq \frac n {|T|^2} \cdot \Paren{\1_T^\top G \1_T - \frac 12 \cdot |T|} = \frac n {|T|^2} \cdot |T| \cdot \Paren{\frac 12 - \Phi_G(T)} \mper
  \]
\end{proof}

\section{Sketch of random-walk rounding for analytically sparse vectors}
\label{sec:sparse-round}

In this section we describe the proof of Theorem~\ref{thm:2-to-4-rounding}.
The key is the following lemma, which says that if $w$ is a vector in the high eigenspaces of $G$ with $\|w\|_4^4 \geq 1/\delta n$, then there is a level set $S$ of $|w|$ (applying the absolute value function coordinate-wise) containing $O(\delta n)$ coordinates such that $\|\frac 1 {|S|} \sum_{i \in S} b_i\|$ is large.
The rest of the proof follows the same argument in Lemma~\ref{lem:spectral-ub}, showing that as the random walk is run from initial distribution over $S$, it must eventually encounter a small cut with imperfect expansion or else it would violate the local Cheeger inequality.

\begin{lemma}
Let $G$ a graph with isotropic spectral embedding $b_i, \ldots, b_n$ and corresponding uniform distribution $D$.
Suppose there exists a unit vector $v$ so that $\E\iprod{v, X}^4 \geq \frac{1}{\delta}$.
Then, there exists $t > 0$ so that if we let $S = \{i : \Abs{\iprod{v, X}} > t \}$, then $|S| \leq O(\delta n)$ and $\Abs{\E \iprod{v, X} | E} \geq \Omega \left( \frac{1}{\delta^{1/2} \log^{1/4} 1 / \delta} \right)$.
\end{lemma}
\begin{proof}
By Fact~\ref{fact:means-to-moments}, we know that for all $p > 4$ there exists some event $A$ so that 
\begin{equation}
\label{eq:tail-level-set}
\Pr[A]^{4/p} \Abs{\E \iprod{v, X} | A}^4 \cdot \frac{p}{p - 4} \geq \frac{1}{\delta} \; .
\end{equation}
Let $A_p$ be the set which achieves the largest value for the LHS in~\eqref{eq:tail-level-set}. 
Without loss of generality, we may take $A_p$ to be of the form $A_p = \{i : \Abs{\iprod{v, X}} > t_p \}$ for some $t_p > 0$, since such sets maximize the mean shift in the direction $v$.

% We must have that $t_p \geq ???$, since otherwise 
% \[
% \Pr[A]^{4/p} \Abs{\E \iprod{v, X} | A}^4 \cdot \frac{p}{p - 4} \leq  \; .

% \]
Because $\E \iprod{v, X}^2 = 1$, by Fact~\ref{fact:means-to-moments}, we must have $\Abs{\E \iprod{v, X} | A_p} \leq \frac{1}{\sqrt{\Pr [A_p]}}$.
Thus, $\Pr [A_p] \leq \delta^{p / (2p - 4)}$, as otherwise we would have
\[
\Pr[A_p]^{4/p} \cdot \Abs{\E \iprod{v, X} | A_p}^4 &\leq \frac{1}{\Pr [A_p]^{2 - 4/p}} < \frac{1}{\delta} \; ,
\]
which contradicts our choice of $A_p$.
This implies that for all $p > 4$, we have
\[
\delta^{4/(2p - 4)} \Abs{\E \iprod{v, X} | A_p}^4 \cdot \frac{p}{p - 4} \geq \frac{1}{\delta}
\]
For $p \leq 6$, if we let $q = p - 2$ we have that
\begin{align*}
\delta^{4/(2p - 4)} \frac{p}{p - 4} \geq \delta^{2/q} \frac{4}{q - 2} \; ,
\end{align*}
so optimizing over $q > 2$ and using Fact~\ref{fact:s-to-r-bound} (see below) yields that by choosing $q = 2 \frac{\log 1 / \delta}{\log 1/\delta - 1}$, we obtain that
\[
\Abs{\E \iprod{v, X} | A_{q + 2}}^4 \geq \frac{O(1)}{\delta^2 \log 1 / \delta} \; .
\]
Finally, in this case, we have
\[
\Pr A_{q + 2} \leq \delta^{1 + 1/(2 \log 1 / \delta)} = O(\delta) \; .
\]
This completes the proof of the lemma.
\end{proof}

\begin{fact}
\label{fact:s-to-r-bound}
Let $x \in (0, 1)$, and let $r \geq 2$.
Then we have
\[
\min_{s > r} x^{r/s} \frac{s}{s - r} \leq e x \log \frac{1}{x} \; ,
\]
and the minimum is attained at $s = r \cdot \frac{\log 1/x}{\log 1 / x - 1}$.
\end{fact}
\begin{proof}[Proof of Fact~\ref{fact:s-to-r-bound}]
By monotonicity of logarithm, it suffices to find the minimizer of the function
\[
f(s) = \frac{r}{s} \log x + \log s - \log (s - r) \; .
\]
Taking derivatives, we find that
\[
f'(s) = - \frac{r}{s^2} \log x + \frac{1}{s} - \frac{1}{s - r} \; .
\]
Thus solving for $f' (s) = 0$, the minimizer of $f$ must satisfy
\[
\frac{r}{s^2} \log \frac{1}{x} = \frac{r}{s (s - r)} \; ,
\]
or equivalently $s / (s - r) = \log 1 / x$ and $r/s = 1 - 1/\log (1/x)$.
Plugging these bounds into the original function yields the desired estimate.
\end{proof}

%% file: content/local-cheeger.tex
%!TEX root = ../main.tex

\subsection{Proofs of Local Cheeger Inequalities}

\begin{lemma}[Local Cheeger Inequality \cite{steurer-thesis}]
\label{lem:local-cheeger}
  For every $v \in \R^n$ there is a level set $S \subseteq V$ of the vector $w_i = v_i^2$ with $|S| \leq \delta n$ and expansion
  \[
  \Phi_G(S) \leq \frac{\sqrt{1 - \iprod{v, Gv}^2/ \|v\|^4}}{1 - \|v\|_1^2 / \delta n \|v\|^2}\mper
  \]
\end{lemma}

\begin{proof}[Proof of Lemma~\ref{lem:local-cheeger-2}]
    We follow the proof in \cite{d-to-1-steurer}, keeping track of a factor of $1/\e^2$ missing in that proof; at the end we apply a standard sub-sampling reduction used in e.g. \cite{DBLP:conf/coco/RaghavendraST12}.

    First, dividing by $\|f\|_1$, we may assume that $\|f\|_1 = 1$; i.e. that $f$ is a probability vector.
    We will apply Lemma~\ref{lem:local-cheeger} to the distribution $g = (f + Gf)/2$.
    Clearly $\|g\|_1 = 1$.

    Since $G$ is contractive in $2$-norm, we have $\|g\|_2 \leq \|f\|_2$.
    But since $f, Gf$ are nonnegative, also $\|g\|_2 \geq \|f\|_2/2$.

    Finally, consider
    \[
    \iprod{g,Gg} = \iprod{f, Gf} + 2 \iprod{f, G^2 f} + \iprod{f, G^3 f} \geq 2 \|Gf\|^2 \geq 2 \e \|f\|^2
    \]
    where we used that $\iprod{f, G^3 f}, \iprod{f, Gf} \geq 0$ by nonnegativity, and we used our hypothesis on $\|Gf\|^2$.
    Plugging these bounds into Lemma~\ref{lem:local-cheeger}, we find that there is a level set $S$ of $g$ having size at most $\delta n / (\gamma \e^2)$ such that
    \[
    \Phi_G(S) \leq \frac{\sqrt{ 1 - \iprod{g, Gg}^2 / \|g\|^4}}{1 - 20 \e^2 /(\delta n \|g\|^2)} \leq \frac{\sqrt{ 1- 4 \e^2}}{1- 20\e^2} \leq 1 - \Omega(\e^2)\mper
    \]

    Let $T$ be a random subset of $S$ of size $\delta n$.
    A simple computation shows that
    \[
      \E_{T} (1 - \Phi_G(T)) \geq \gamma \e^2 (1 - \Phi_G(S)) \geq \Omega(\gamma \e^4)\mper 
    \]
    as claimed.
\end{proof}

\begin{proof}[Proof of Lemma~\ref{lem:local-cheeger-2.5}]
We begin by proving the statement prior to the ``moreover,'' then we describe how the proof may be slightly altered in the case that $G$ contains a small non-expanding set.

Our proof proceeds very similarly to the proof in~\cite{steurer-thesis}.
Let $c, C$ be constants to be determined later.
Let $T_t$ be a random subset drawn from the following distribution: first, $t$ is drawn uniformly from $[0, 1]$, then $T_t = \{ i \in [n]: g_i^2 \geq t \}$.
We first establish a number of properties of this distribution.
Observe that since $|f_i| \leq 1$ for all $i$, then since $G$ is a random walk matrix, $|G f_i| \leq 1$ for all $i$ as well, and so $g_i^2 \leq 1$ for all $i$.
Therefore, by a simple calculation, we have that
\[
\E_t \Brac{\Abs{T_t}} = \sum_{i \in [n]} g_i^2 = \Norm{g}^2 \; .
\]
We also have that $\Norm{g} ^2 \geq (1 - \eta)^2 \Norm{ f}^2 = (1 - \eta)^2 \delta n$.
Moreover, if $t \leq (1 - \eta)^2$, $S \subseteq T_t$ and hence $|T_t| \geq \delta n$.
Therefore 
\begin{equation}
\label{eq:lc-prob-lb}
\Pr_t \Brac{|T_t| < \delta n} \leq 2 \eta \; .
\end{equation}
For any $U, V \subseteq [n]$, let $G(U, V) = \Pr_{(i, j) \sim G} [i \in U, j \in V]$ be the fraction of edges going from $U$ to $T$, so that $\Phi_G (U) = n G(U, [n] \setminus U) / |U|$. 

Then, by the same calculations as those done in~\cite{steurer-thesis}, we still have the following three inequalities:
\begin{align}
\E_t \Brac{|T_t|^2} &\leq \Norm{g}_1^2 \; , \label{eq:lc-l2-bound}\\
\E_t \Brac{|T_t| \1_{|T_t| > C \delta n / (\eta \eps)^2}} &\leq \frac{\eta^2 \eps^2}{C \delta n} \E_t \Brac{|T|^2} \; , \label{eq:lc-trunc-bound} \\
n \cdot \E_t G(T_t, [n] \setminus T_t) &\leq \Norm{g}^2 \sqrt{1 - \iprod{g, G g}^2 / \Norm{g}^4} \label{eq:lc-expansion-bound} \; .
\end{align}
We now specialize each of these three inequalities to our setting.
Observe that $f$ is nonnegative and satisfies $\Norm{f}_1 = |S| = \delta n$, and so because $G$ is a random walk matrix, we have $\Norm{g}_1 = |S|$, and so~\eqref{eq:lc-l2-bound} simply becomes   
\[
\E_t \Brac{|T_t|^2} \leq (\delta n)^2 \; .
\]
Plugging this bound into~\eqref{eq:lc-trunc-bound} yields that 
\begin{equation}
\label{eq:lc-trunc-bound2}
\E_t \Brac{|T_t| \1_{|T_t| > C \delta n / (\eta \eps)^2}} \leq \frac{\eps^2 \eta^2}{C} \delta n \; .
\end{equation}
Finally, observe that
\begin{align}
\label{eq:lc-quadform-lb}
\iprod{g, Gg} &= (1 - \eta)^2 \iprod{f, Gf} + 2 \eta (1 - \eta) \Norm{G f}^2 + \eta^2 \iprod{f, G^3 f} \nonumber \\
&\stackrel{(a)}{\geq} 2 \eta (1 - \eta) \Norm{G f}^2 \nonumber \\
&\stackrel{(b)}{\geq} 2 \eta (1 - \eta) \eps \Norm{f}^2 \; ,
\end{align}
where (a) follows from the nonnegativity of $f$, and (b) follows from assumption.
Moreover $\Norm{g}^2 \leq \Norm{f}^2 = \delta n$ since $G$ is contractive in $\ell_2$.
Thus~\eqref{eq:lc-expansion-bound} simplifies in our setting to give
\begin{equation}
\label{eq:lc-expansion-2}
n \cdot \E_t  G(T_t, [n] \setminus T_t) \leq \Norm{g}^2 \sqrt{1 - (2 \eta (1 - \eta) \eps)^2} = \Norm{g}^2 (1 - \Omega(\eta^2 \eps^2)) \; .
\end{equation}

Now, let $T^*$ be the level set of $g^2$ with size in the range $I = [c \eta \eps^2 \delta n, C \delta n / (\eta \eps)^2]$ with minimal $\Phi (T^*)$.
Since $\Phi (T) = n \cdot \frac{G(U, [n] \setminus U)}{|U|}$, we have
\begin{align*}
\Phi (T^*) &\leq n \frac{\E_t G(T_v, [n] \setminus T_t)}{\E_t |T_t| \1_{|T_t| \in I}} \\
&\leq n \frac{\E_t G(T_v, [n] \setminus T_t)}{\E_{t} |T_t|\1_{|T_t| \in I}} \\
&= n \frac{\E_t G(T_v, [n] \setminus T_t)}{\E_{t} |T_t| - \E_{t} |T_t| \1_{|T_t| < c \eta \eps^2 \delta n} - \E_t |T_t| \1_{|T_t| > C \delta n / (\eta \eps)^2}} \\
&\stackrel{(a)}{\leq} n\frac{\E_t G(T_v, [n] \setminus T_t)}{\E_{t} |T_t| - c' \eta^2 \eps^2 \delta n} \\
&\stackrel{(b)}{\leq} \frac{\Norm{g}^2 (1 - O(\eta^2 \e^2))}{\Norm{g}^2(1 - 2 \eps^2 \eta^2)} \\
&\leq 1 - O(\eta^2 \eps^2) \; ,
\end{align*}
for an appropriate choice of $c$ sufficiently small and $C$ sufficiently large.
Here (a) follows from~\eqref{eq:lc-trunc-bound2}, and since 
\begin{equation}
\label{eq:lc-not-too-small}
\E_{t} |T_t| \1_{|T_t| < c \eta \eps^2 \delta n} \leq c \eta \eps^2 \delta n \Pr \Brac{|T_t| < c \eta \eps^2 \delta n} \leq c \eta^2 \eps^2 \delta n
\end{equation}
by~\eqref{eq:lc-prob-lb}, and (b) follows since $\| g \|_2^2 \geq (1 - \eta)^2 \| f \|^2 = (1 - \eta)^2 \delta n$.
This completes the proof, except for the ``moreover'' statement, proved below.
\end{proof}

\begin{proof}[Proof of Lemma~\ref{lem:local-cheeger-2.5}, ``moreover'' part]
Suppose that $G$ contains a set $R$ as described in the lemma statement.
We describe how the preceeding proof may be altered to ensure that $T \cap R = \emptyset$.

The idea is to replace the function $g$ with the function $g' = \Pi_{\overline R}g$, the projection of $g$ to the coordinates outside $R$.
The random thresholding procedure is applied to the coordinates of $g'$ to produce the set $T$; because $g'$ is supported off of $R$ it holds that $T \cap R = \emptyset$ with probability $1$.

We now verify that properties of $g$ used above also apply to $g'$.
Since $f$ is supported off of $R$, \eqref{eq:lc-prob-lb} continues to hold.
Equations \eqref{eq:lc-l2-bound}, \eqref{eq:lc-trunc-bound}, \eqref{eq:lc-expansion-bound} hold for any choice of $g$ and hence in particular for $g'$.

Because $\|g'\|_1 \leq \|g\|_1$, we obtain \eqref{eq:lc-trunc-bound2} when $T$ is chosen according to the thresholding procedure on $g$.

We need to lower bound $\iprod{g', Gg'}$ to obtain an analogue of \eqref{eq:lc-quadform-lb}.
By expanding, we find
\[
\iprod{g', G g'} = \iprod{g, G g} + 2\iprod{g' - g, Gg'} + \iprod{ g'-g, G(g'-g)}\mper
\]

Because $\Phi_G(R) \leq \e/10$ and $|S| = |R| = \delta n$, we obtain that $\|g'-g\|_1 = \| \Pi_R g\|_1 \leq \eta \e \delta n/10$.
And because as noted before $|g_i| \leq 1$ for all $i$, we have $\| g' \|_\infty \leq 1$.
So $|\iprod{g'-g, Gg'}| \leq \eta \e \delta n$; the same argument applies to $|\iprod{g'-g, G (g'-g)}|$.
And we proved above that $\iprod{g,Gg} \geq 2 \eta(1-\eta) \e \|f\|^2$.
We may assume $\eta \leq 1/2$, so it follows that $\iprod{g', G g'} \geq \eta (1-\eta) \e \|f\|^2$.
Thus up to a factor of $2$, we obtain the analogue of \eqref{eq:lc-quadform-lb} for $g'$ in place of $g$.

Since $\|g'\|^2 \leq \|g\|^2$, we also obtain
\[
n \E_t G(T_t, [n] \setminus T_t) \leq \|g'\|^2 \sqrt{1- (\eta (1-\eta)\e^2)} = \|g\|^2 (1 - \Omega(\eta^2 \e^2))
\]
as in \eqref{eq:lc-expansion-2}.

Finally, since $\Pi_{\overline R} f = f$ (since $f$ is supported off of $R$), it still holds that $\Pr[|T_t| < c \eta \e^2 \delta n] \leq \eta$ as in \eqref{eq:lc-not-too-small}.
The rest of the proof goes through unchanged.
\end{proof}